\documentclass{article}

\usepackage{amsmath,amssymb}
\usepackage{amsfonts,bm}
\usepackage{framed}
\usepackage{braket}
\usepackage{bbm}
\usepackage{graphicx}
\usepackage{subfig}

    \usepackage[final]{neurips_2025}

\usepackage[utf8]{inputenc} 
\usepackage[T1]{fontenc}    
\usepackage{hyperref}       
\usepackage{url}            
\usepackage{booktabs}       
\usepackage{amsfonts}       
\usepackage{nicefrac}       
\usepackage{microtype}      
\usepackage{xcolor}         

\usepackage{natbib}

\usepackage[labelfont=bf]{caption}

\title{Neurons as Detectors of Coherent Sets in Sensory Dynamics}

\begin{document}

\author{%
  \hspace{-1mm}Joshua L. Pughe\mbox{-}Sanford\textsuperscript{1,*}\quad
  Xuehao Ding\textsuperscript{1,*}\quad
  Jason J. Moore\textsuperscript{1,2,*}\quad
  Anirvan M. Sengupta\textsuperscript{1,3}\\[2pt]
  {\bfseries
  \hspace{-6mm}Charles Epstein\textsuperscript{1}\quad
  Philip Greengard\textsuperscript{1}\quad
  Dmitri B. Chklovskii\textsuperscript{1,2}}\\[2pt]
 \hspace{-5mm} \textsuperscript{1}Center for Computational Neuroscience, Flatiron Institute, Simons Foundation, New York, NY, USA\\
  \hspace{-6mm}\textsuperscript{2}Neuroscience Institute, NYU Langone Medical Center, New York, NY, USA\\
  \hspace{-5mm}\textsuperscript{3}Physics Department, Rutgers University, New Brunswick, NJ, USA\\[2pt]
  \hspace{-2mm}\texttt{\{jpughesanford, xding, cepstein, pgreengard, mitya\}@flatironinstitute.org,}\\
  \hspace{-2mm}\texttt{jason.moore@nyulangone.org,}\;
  \texttt{anirvans.physics@gmail.com}\quad
  \textsuperscript{*}\,Equal contribution
}

\maketitle
\vspace{-5mm}
\begin{abstract}
We model sensory streams as observations from high-dimensional stochastic dynamical systems and conceptualize sensory neurons as self-supervised learners of compact representations of such dynamics. From prior experience, neurons learn {\it coherent sets}—regions of stimulus state space whose trajectories evolve cohesively over finite times—and assign membership indices to new stimuli. Coherent sets are identified via spectral clustering of the {\it stochastic Koopman operator (SKO)}, where the sign pattern of a subdominant singular function partitions the state space into minimally coupled regions. For multivariate Ornstein–Uhlenbeck processes, this singular function reduces to a linear projection onto the dominant singular vector of the whitened state-transition matrix. Encoding this singular vector as a receptive field enables neurons to compute membership indices via the projection sign in a biologically plausible manner. Each neuron detects either a {\it predictive} coherent set (stimuli with common futures) or a {\it retrospective} coherent set (stimuli with common pasts), suggesting a functional dichotomy among neurons. Since neurons lack access to explicit dynamical equations, the requisite singular vectors must be estimated directly from data, for example, via past–future canonical correlation analysis on lag-vector representations—an approach that naturally extends to nonlinear dynamics. This framework provides a novel account of neuronal temporal filtering, the ubiquity of rectification in neural responses, and known functional dichotomies. Coherent-set clustering thus emerges as a fundamental computation underlying sensory processing and transferable to bio-inspired artificial systems.
\end{abstract}

Neurons in early sensory areas are traditionally thought to extract from recent inputs low-dimensional latent variables that are maximally informative about the near future \cite{tishby2000information,bialek2006efficient,palmer2015predictive,chalk2018toward}. Such extraction exploits statistical regularities acquired over evolutionary, developmental, and behavioral timescales from previously encountered natural stimuli \cite{attneave1954some,barlow1961possible,simoncelli2001natural}. To formalize this intuition for temporally correlated sensory stimuli, we postulate that they are generated by high-dimensional, potentially nonlinear, stochastic dynamical processes, and conceptualize neurons as self-supervised learners of \emph{coherent sets}—regions of the stimulus state space that evolve cohesively over finite time intervals~\cite{dellnitz1999approximation,froyland2013analytic}—thus enabling compact representations of sensory dynamics.

Coherent sets can be uncovered via spectral clustering of the \emph{stochastic Koopman operator (SKO)}—a linear, albeit infinite‑dimensional, operator that evolves observables over a finite time interval \cite{froyland2013analytic,klus2024dynamical}. The sign of the first non‑trivial (subdominant) singular function of the SKO partitions state space into two minimally interacting coherent sets (Fig. 1a). Accordingly, a neuron can compute a membership index of a new input by evaluating the sign of a subdominant singular function. Because singular values and functions remain real even for irreversible dynamics, this approach generalizes metastable set detection beyond the reversible cases that eigenfunction methods require \cite{wu2020variational}. 

We demonstrate that for multivariate Ornstein–Uhlenbeck (OU) process \cite{enwiki:1286449686}—a canonical example of linear stochastic dynamics and a reasonable model of summed input to a neuron \cite{tuckwell2012stochastic,DestexheOU_2001}—a subdominant singular function of the SKO corresponds to a projection of the input onto a singular vector of the whitened finite-time transition matrix, Fig. 1b. A neuron that stores this singular vector in its synaptic weights and temporal filter can compute the corresponding membership index via the sign of the weighted sum of the inputs. Although the conventional spectral clustering framework assumes stability, our singular function results extend to unstable systems by focusing on singular values closest to one. In addition to projecting on the right singular vector \emph{predicting} near-future inputs (Fig. 1b), projecting on the left one is also possible---\emph{retrospecting} the recent past (Fig. 1c). Because such projections require different synaptic weights and temporal filters they must be implemented by distinct neurons.

Recognizing that the underlying dynamical equations are not available to neurons,  a biologically plausible detection of coherent sets requires a data-driven algorithm that can infer them directly from observations. This can be done using past-future canonical correlation analysis (CCA) \cite{koltai2018,klus2024dynamical} which can be implemented locally~\cite{lipshutz2021biologically}. This algorithm has the additional advantage of being applicable to nonlinear dynamics: it relies on estimating a Galerkin projection of the SKO onto a chosen functional basis via Monte Carlo integration over observed data. If the upstream neurons implement such projection, the post-synaptic neuron could then locally learn a requisite singular vector and compute a membership index via the sign of the weighted sum of the inputs. 

Viewing neurons as coherent set detectors sheds light onto several longstanding neurophysiological observations. First, temporal receptive fields of neurons emerge naturally as subdominant singular vectors projecting input lag‑vector representations of dynamical states. Second, the ubiquity of response rectification in neurons, exemplified by the well‑known ON/OFF segregation in early visual circuits, is interpreted as a principled clustering mechanism. Finally, the theory predicts complementary neuronal classes that predict near future using predictive coherent sets or retrospect recent past using retrospective ones, consistent with known neuronal functional dichotomies such as tufted versus mitral cells in the olfactory bulb or non‑lagged versus lagged cells in the lateral geniculate nucleus (LGN). Thus, the detection of coherent sets can serve as a powerful algorithmic primitive for neural computation supporting prediction and retrospection. This offers insights into biological processes and could inspire future artificial neural networks.

The remainder of the paper is organized as follows. Section 1 discusses related work. Section 2 reviews the definition of transfer operators, coherent sets and their connection to spectral clustering through the singular functions of the SKO. In Section 3, we derive the central result: under OU dynamics, the subdominant singular functions of the SKO correspond to dot products between the state and the singular vectors of the whitened transition matrix. Section 4 reviews a data-driven algorithm for identifying subdominant singular functions, which extends naturally to nonlinear dynamics and defines neuronal units rectifying positive or negative parts of the subdominant singular function. Section 5 analyzes and reviews several experimental datasets, interpreting temporal receptive fields through the hypothesis that biological neurons cluster coherent sets.
\vspace{-5mm}
\begin{figure}[h]
  \centering
  \subfloat{\includegraphics[height=0.2\textwidth]{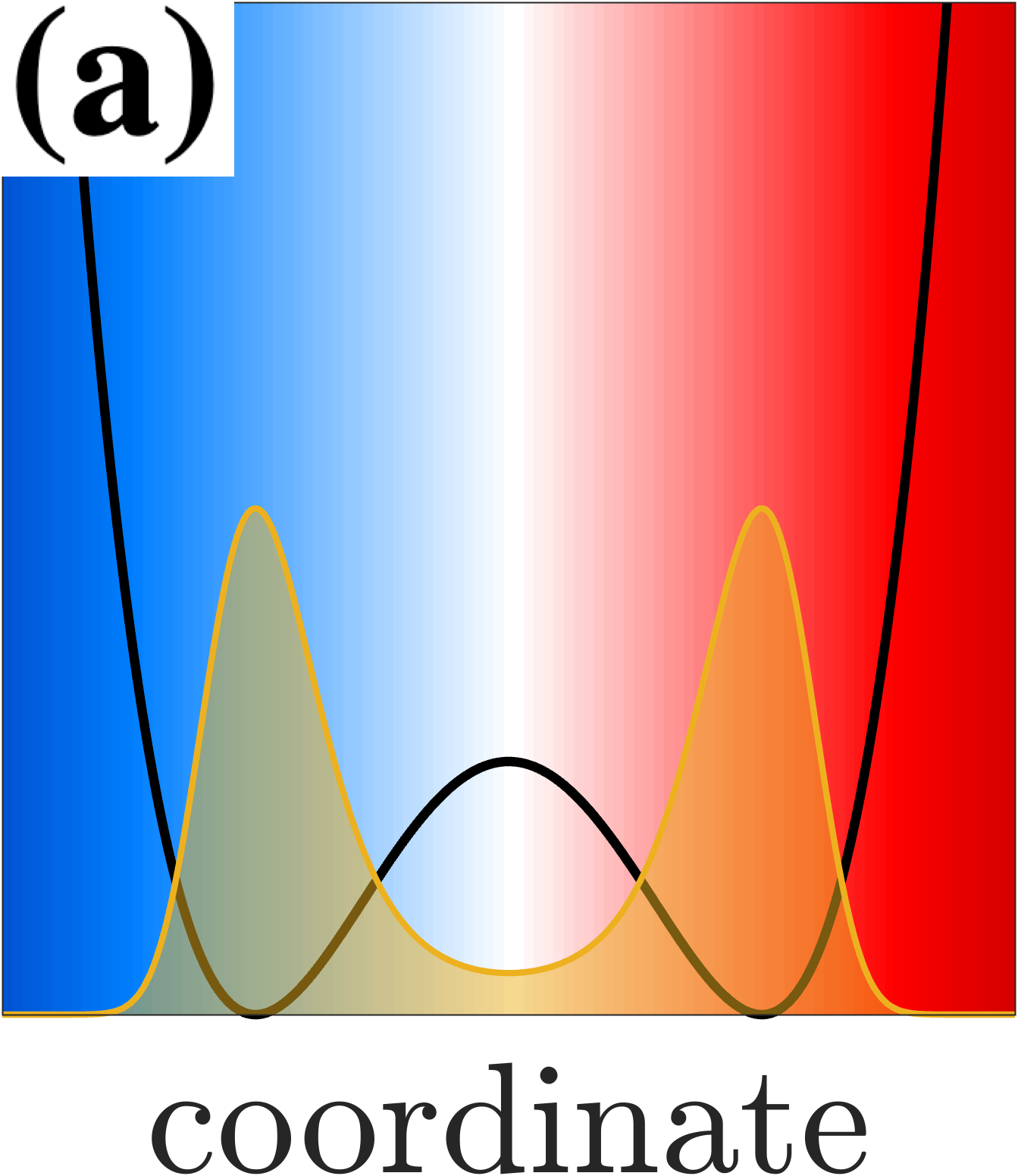}}\hspace{5mm}
  \subfloat{\includegraphics[height=0.2\textwidth]{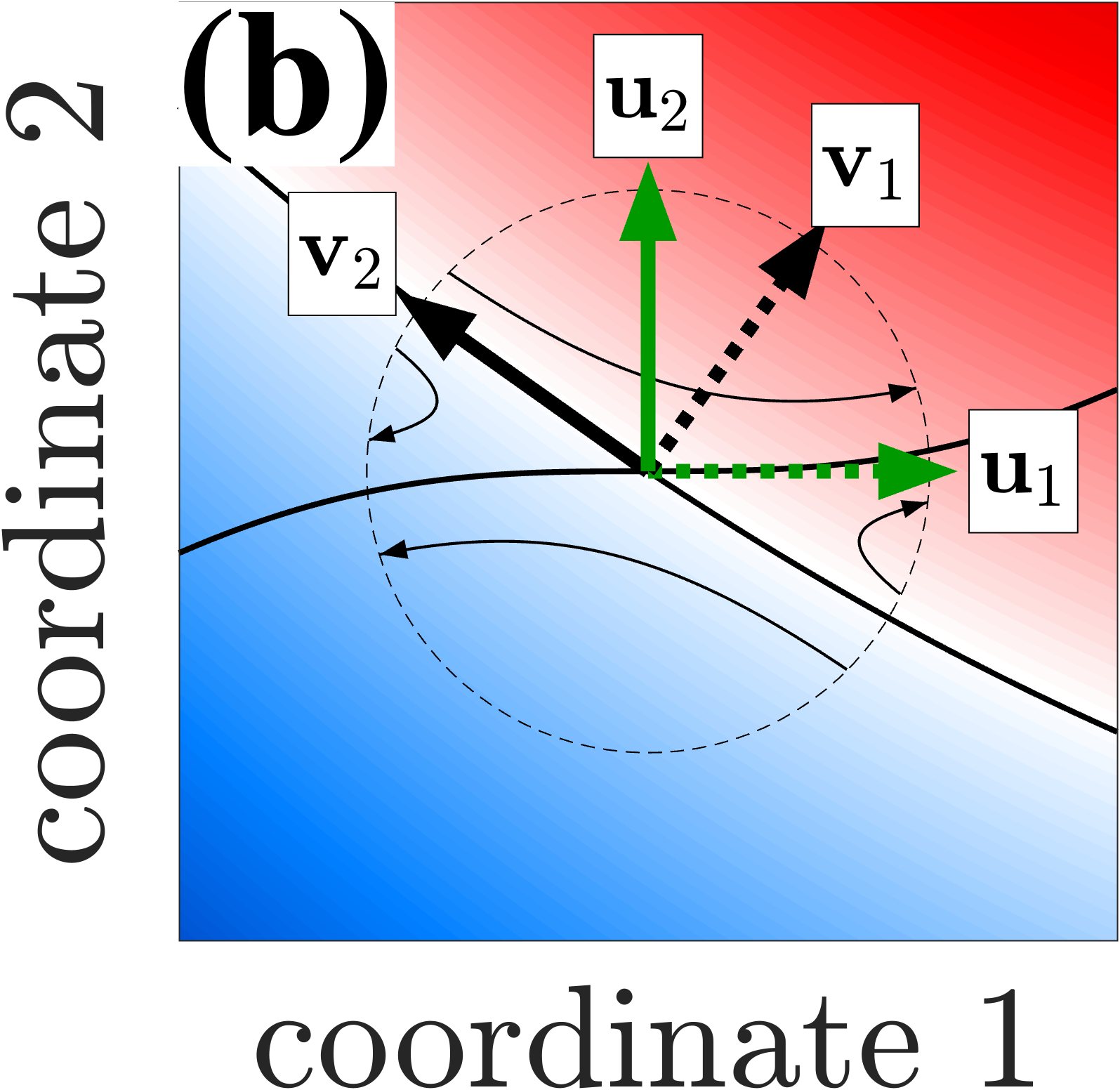}}\hspace{5mm}
  \subfloat{\includegraphics[height=0.2\textwidth]{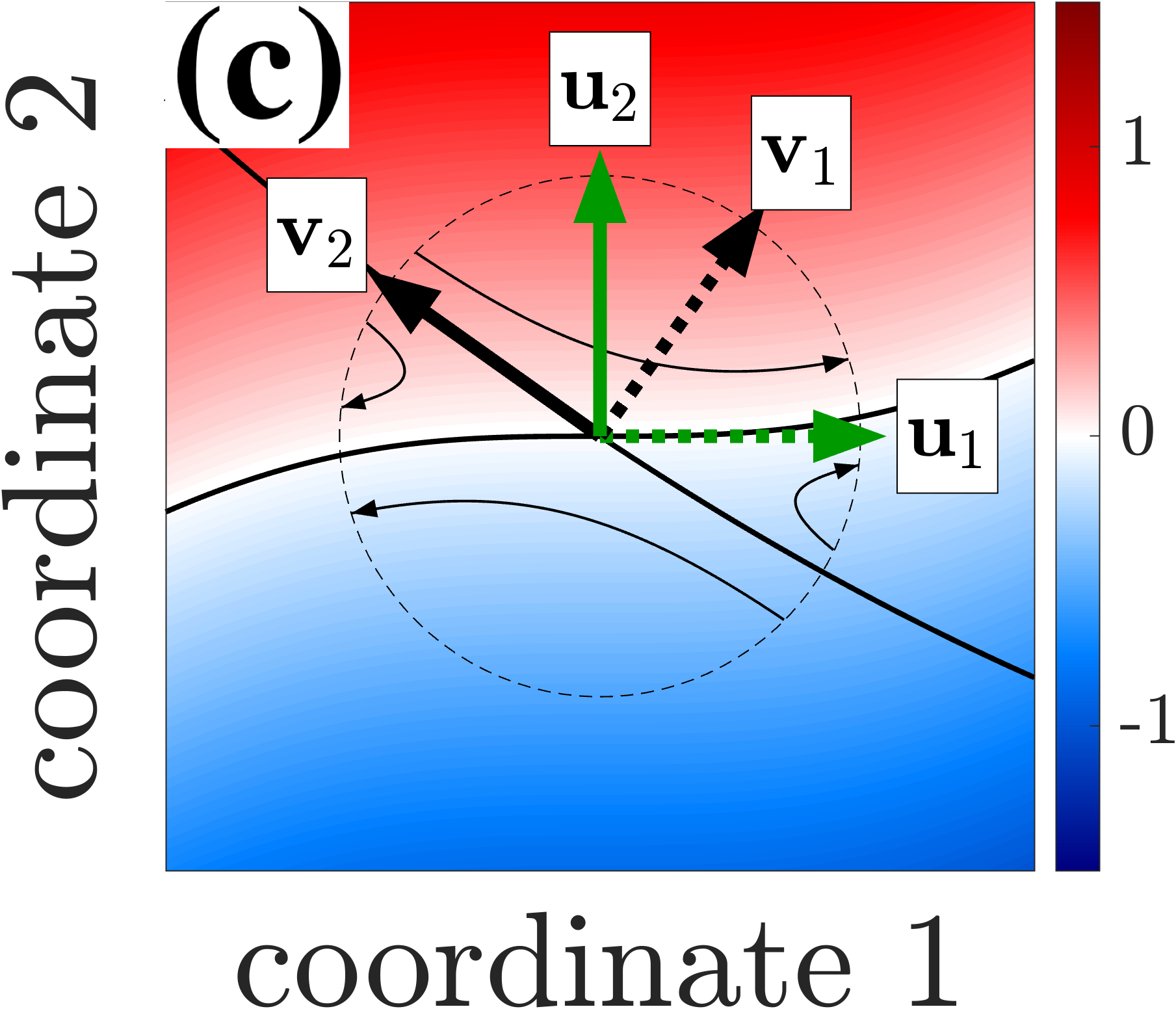}}
  \caption{Subdominant SKO singular functions partition states into coherent sets (blue vs red). {\bf (a)} For a one-dimensional double-well potential (black) the dominant singular function corresponds to a stationary distribution (yellow) and the subdominant singular function (red-blue) partitions the state space to minimize the leakage between the coherent sets. Left ({\bf b}) and right ({\bf c}) subdominant singular functions partition the state space near a 2D saddle point based on shared future and past, respectively. Black lines indicate attractive and repulsive invariant manifolds approximated by the linear subspaces in the vicinity of the saddle point (circle). Left (green) and right (black) singular vectors of the finite-time whitened transfer matrix, $[{\bf u}_1 {\bf u}_2]{\bf \Lambda}[{\bf v}_1 {\bf v}_2]^\top $. As the forecast horizon grows, ${\bf v}_1$ becomes orthogonal to the stable subspace and ${\bf u}_2$ becomes orthogonal to the unstable subspace.
  }
\end{figure}

\section{Related Work}

Viewing early sensory processing as efficient or predictive coding of natural stimuli has a long tradition~\cite{attneave1954some,barlow1961possible,simoncelli2001natural,price2022efficient,bialek2006efficient}. Closest to our work is the information bottleneck (IB) framework~\cite{tishby2000information,bialek2006efficient,palmer2015predictive,chalk2018toward}, which for Gaussian variables reduces to past–future CCA~\cite{chechik2003information}. Extending IB to dynamical systems and restricting compression to one bit corresponds to encoding the sign of the subdominant singular function of the SKO. Compared to slow feature analysis~\cite{wiskott2002slow,lipshutz2020biologically}, we incorporate an explicit dynamical-systems formulation and introduce a natural nonlinearity that enables hierarchical architectures without hand-crafted features.  Predictive and retrospective neuron classes have appeared in a lattice filter model of the visual pathway~\cite{gregor2012lattice}, though that work was limited to linear processing. Relative to predictive coding models~\cite{srinivasan1982predictive,rao1999predictive,druckmann2012mechanistic, gregor2012lattice}, where neurons emit prediction errors, our neurons encode coherent-set memberships. While the ON/OFF division has been attributed to metabolic efficiency~\cite{gjorgjieva2014benefits}, we propose a computational explanation. Our account of lagged and non-lagged LGN cells differs from earlier models~\cite{dong1995temporal} by not relying on nonlinearity, suggesting that these types can emerge alongside ON/OFF segregation.

Prediction and retrospection in saddle-point dynamics have been discussed in the context of unstable periodic orbits of chaotic attractors using the dominant mode of the local SKO (and its adjoint)~\cite{cvitanovic2012knowing,heninger2015neighborhoods}, whereas we focus on a subdominant singular function defining a coherent set pair~\cite{froyland2010transport,froyland2013analytic}. Saddle point analysis based on eigenfunctions instead of singular functions predicted neuronal filters orthogonal to growing exponentials, yielding predictive neurons without considering retrospection~\cite{golkar2024neuronal}.

Clustering has been previously proposed as a model of static neuronal computation on temporally uncorrelated inputs~\cite{pehlevan2017clustering}, capturing rectification and sparsity but not temporal receptive fields or sensory dynamics. Probabilistic coding frameworks such as the Bayesian brain hypothesis~\cite{ma2006bayesian} have been widely studied although we are not aware of the suggestion that single neurons represent eigenfunctions or singular functions of the SKO. Self-supervised learning has been applied to visual networks~\cite{singer2018sensory}, but these models typically omit analysis on the neuronal level. Koopman-based objectives have recently been incorporated into deep architectures to extract predictive features~\cite{choi2024koopman}. In contrast, we use SKO singular functions to define coherent sets~\cite{froyland2010transport,froyland2013analytic}, providing a direct explanation for rectification and a principled division into predictive and retrospective neurons.

Estimation of singular functions from data originated in molecular dynamics as the VAMP framework~\cite{wu2020variational}. VAMPnets~\cite{mardt2018vampnets} learn such functions with deep ReLU networks trained by backpropagation. In our model, features arise from rectified projections onto singular vectors within each neuron and can be hierarchically composed without backpropagation.

\section{Discovering coherent sets via spectral clustering of transfer operators}

We begin with the concept of coherent sets, which are regions of state space that tend to move as a whole under the stochastic dynamics over finite time intervals~\cite{froyland2010transport, froyland2013analytic}.  Formally,
a pair of sets $(\mathcal{A},\mathcal{B})$ is coherent if a state that begins in $\mathcal{A}$ at time $t$ is very likely to be found in $\mathcal{B}$ at time $t+\tau$:
$$\mathbb{P}[{\bf X}(t+\tau)\in \mathcal{B}\ |\ {\bf X}(t)\in \mathcal{A}]\approx 1,$$
where $\mathbb{P}[{ X} | {Y}]$ is a probability of ${X}$ given ${Y}$, ${\bf X}(t)\in \mathcal{X}$ denotes the random state of the system at time $t$. If in addition the probabilities of visiting $\mathcal{A}$ at time $t$ and $\mathcal{B}$ at time $t+\tau$ are equal, then this relation works in both directions: observing the system in $\mathcal{B}$ at time $t+\tau$ makes it likely that it was in $\mathcal{A}$ at time $t$. Thus, membership in coherent sets, $\mathcal{A}$ and $\mathcal{B}$, can be used to predict the future and retrospect the past, respectively, of a stochastic process. 

To compute coherent sets we utilize the stochastic Koopman operator (SKO), which encodes how \emph{observables} evolve in expectation~\cite{froyland2010transport, froyland2013analytic,klus2024dynamical}, see the Supplement, Section 1 for a concise introduction. An observable is any scalar function of ${\bf X}$ with a finite expectation value over the stochastic process. For instance, from the neuron's perspective, each synaptic input can be thought of as such an observable of stimuli. The SKO acts on an observable $f$
\begin{equation}
    \mathcal{K}_\tau f({\bf x})= \mathbb{E}[f({\bf X}(t+\tau))| {\bf X}(t) = {\bf x}],
\end{equation}
i.e. it maps the present observable to its expected value $\tau$ units of time later. Linearity of the expectation makes $\mathcal{K}_\tau$ a linear operator.

Let $\mu({\bf x})$ and $\nu({\bf x})$ denote the probability densities of the system’s state at times $t$ and $t+\tau$, respectively~\cite{froyland2010transport, froyland2013analytic,klus2024dynamical}. For any measurable region $\mathcal{A}\subseteq \mathcal{X}$,
\begin{align}
\mathbb{P}[{\bf X}(t) \in \mathcal{A}] &= \int_{\mathcal{A}} \mu({\bf x})\ d{\bf x}, &
\mathbb{P}[{\bf X}(t+\tau) \in \mathcal{A}] &= \int_{\mathcal{A}} \nu({\bf x})\ d{\bf x}.
\end{align}
These densities define inner products between functions $f$ and $g$ as
\begin{align}
\langle f, g \rangle_\mu &= \int_{\mathcal{X}} f({\bf x})^* g({\bf x})\, \mu({\bf x})\, d{\bf x}, &
\langle f, g \rangle_\nu &= \int_{\mathcal{X}} f({\bf x})^* g({\bf x})\, \nu({\bf x})\, d{\bf x}.
\end{align}
When the state space is divided into two pairs of regions, identifying the most coherent ones amounts to maximizing the following objective \cite{froyland2013analytic,klus2024dynamical}:
\begin{equation}
\label{eq:quotient}
\begin{aligned}
&\max_{\mathcal{A},\,\mathcal{B}} \Big\{
\mathbb{P}\left[{\bf X}(t+\tau)\in \mathcal{B}\,\Big|\,{\bf X}(t)\in \mathcal{A}\right]
+
\mathbb{P}\left[{\bf X}(t+\tau)\in \mathcal{B}^c\,\Big|\,{\bf X}(t)\in \mathcal{A}^c\right]
\Big\} \\
&= \max_{\mathcal{A},\,\mathcal{B}} \left\{
\frac{\langle {\bf 1}_{\mathcal{A}}\,, \mathcal{K}_\tau {\bf 1}_{\mathcal{B}} \rangle_\mu}
{\langle {\bf 1}_{\mathcal{A}}\,, {\bf 1}_{\mathcal{A}} \rangle_\mu}
+
\frac{\langle {\bf 1}_{\mathcal{A}^c}\,, \mathcal{K}_\tau {\bf 1}_{\mathcal{B}^c} \rangle_\mu}
{\langle {\bf 1}_{\mathcal{A}^c}\,, {\bf 1}_{\mathcal{A}^c} \rangle_\mu}
\right\},\\
&\text{subject to }\int_\mathcal{A} \mu({\bf x})\ d{\bf x} = \mathbb{P}[{\bf X}(t)\in \mathcal{A}]=\mathbb{P}[{\bf X}(t+\tau)\in \mathcal{B}] = \int_\mathcal{B} \nu({\bf x})\ d{\bf x},
\end{aligned}
\end{equation}
where $(\mathcal{A},\mathcal{B})$ are measurable subsets of the state space $\mathcal{X}$, and $(\mathcal{A}^c,\mathcal{B}^c)$ denote their complements. The indicator function ${\bf 1}_{\mathcal{A}}({\bf x})$ equals $1$ if ${\bf x}\in \mathcal{A}$ and $0$ otherwise. Large values of the objective correspond to pairs of regions that remain coherent under the dynamics---that is, regions that are least dispersive over the time interval $\tau$.

In practice, the maximizers of this quotient are well-approximated by certain singular functions of the SKO~\cite{froyland2010transport, froyland2013analytic,klus2024dynamical}. To see this, note that $\mathcal{K}_\tau$ has an adjoint operator $\mathcal{K}_\tau^\dagger$, defined so that
\begin{align}\label{eq:innerprod}
\langle f, \mathcal{K}_\tau g \rangle_\mu &= \langle \mathcal{K}_\tau^\dagger f, g \rangle_\nu.
\end{align}
The finite-time forward-backward and backward-forward operators 
\begin{align}
\label{FB}
    \mathcal{F}_\tau &= \mathcal{K}_\tau \mathcal{K}_\tau^\dagger, & \mathcal{B}_\tau &= \mathcal{K}_\tau^\dagger \mathcal{K}_\tau
\end{align}
are then self-adjoint and, under mild assumptions, compact \cite{klus2024dynamical,koltai2018}. As such, each admits a countable spectral decomposition in terms of orthonormal basis
\vspace{-2mm}
\begin{align}
    \mathcal{F}_\tau &= \sum_{i=0}^D \lambda_i^2 \, v_i \,\langle v_i , \cdot \rangle_\mu, &
    \mathcal{B}_\tau &= \sum_{i=0}^D \lambda_i^2 \, u_i \,\langle u_i , \cdot \rangle_\nu,
\end{align}
where $D$ may be infinite. The functions $u_i({\bf x})$ and $v_i({\bf x})$ are the singular functions of $\mathcal{K}_\tau$, satisfying 
\begin{align}
     [\mathcal{K}_\tau u_i]({\bf x})&=\lambda_i { v}_i({\bf x}), &  [\mathcal{K}_\tau^\dagger v_i]({\bf x})&=\lambda_i { u}_i({\bf x}).
\end{align}
Pairs of singular functions with large $\lambda_i$ are approximate maximizers of \eqref{eq:quotient}. Moreover, the sets 
\begin{align}
\mathcal{V}^{\pm}_i &= \{ {\bf x} \; : \; \pm v_i({\bf x}) > 0 \}, &
\mathcal{U}^{\pm}_i &= \{ {\bf x} \; : \; \pm u_i({\bf x}) > 0 \},
\end{align}
approximate coherent sets, where $\mathcal{U}^{\pm}_i$ is approximately the image of $\mathcal{V}^{\pm}_i$ under the dynamics.

The principal singular functions of the SKO are trivial, ${u}_0({\bf x}) = { v}_0({\bf x})=1$: they define the trivial coherent set, $\mathcal{U}_0=\mathcal{V}_0=\mathcal{X}$. The subdominant singular functions define the least dispersive non-trivial coherent set pair of the dynamics, $(\mathcal{V}^{\pm}_i,\mathcal{U}^{\pm}_i)$. We propose that the neurons compute membership indices by evaluating the signs of such singular functions. Left singular functions look forward in time predicting the future. Right singular functions look backward in time allowing to retrospect past events. Both types of measurements are important in ascertaining the state of partially observed dynamical systems, and we predict the existence of both predictive and retrospective neurons.

\section{Koopman singular functions for Ornstein-Uhlenbeck processes}

Identifying coherent sets through the subdominant singular functions of the SKO provides a rigorous theoretical framework, whose application requires computing these singular functions. In this section, we attempt to provide an intuitive picture by considering a linear stochastic dynamical system --- OU process --- for which subdominant SKO singular functions can be found in the closed form providing much needed intuition. The temporal dynamics of neuronal input currents elicited by sensory stimuli has long been approximated by the OU process~\cite{tuckwell2012stochastic,DestexheOU_2001}:
\begin{equation}
    \frac{d\mathbf{x}(t)}{dt}=\mathbf A\mathbf x(t)+\mathbf \xi(t),
\end{equation}
where $\langle\mathbf \xi(t) \mathbf \xi(t')^\top\rangle=\mathbf D\delta(t-t')$. Here, we assume that $\bf A$ is a real matrix with real eigenvalues and eigenvectors. We are particularly interested in saddle-point OU for the following reasons. If a critical point is purely repulsive, it will not be visited by the autonomous dynamics and, hence, will be physically irrelevant. Applying this approach to nearly isotropic attractive critical points gives partitions that are not ‘distinguished’~\cite{froyland2013analytic}, i.e. do not represent a genuine clustering of states which could describe qualitatively different parts of the phase space.

The probability density of the state variable $\bf{x}$, $p({\bf x}, t)$, evolves according to the forward Kolmogorov equation (aka the Fokker-Planck equation) utilizing the forward Kolmogorov operator $\mathcal{L}$ \cite{pavliotis2014stochastic}:
\begin{equation}
\label{FP}
\frac{\partial p({\bf x},t)}{\partial t}= \mathcal{L}\, p({\bf x},t) \equiv -\nabla\cdot \left({\bf A x} \, p({\bf x},t) \right) + \frac{1}{2} \nabla\cdot ({\bf D}\nabla p({\bf x}, t)).
\end{equation}
The dynamics of measurement expectation $g({\bf x},t) = \mathbb{E}[g({\bf X}(t))|{\bf X}(0) = x]$ are given by the backward Kolmogorov equation: 
\begin{equation}
\frac{\partial g({\bf x},t)}{\partial t}=\mathcal{L}^\dagger g({\bf x}, t) \equiv ({\bf A x})\cdot\nabla \, g({\bf x},t) + \frac{1}{2} \nabla\cdot ({\bf D}\nabla g({\bf x}, t)),
\end{equation}
where $\mathcal{L}^\dagger$ is the adjoint of the forward Kolmogorov operator with respect to the standard Euclidean inner product, serving as the generator of the SKO:
\begin{align}
\mathcal{K}_\tau = \exp{(\mathcal{L}^\dagger\tau)}.
\end{align}
The stationary distribution of probability density under $\mathcal{L}$, satisfying $\mathcal{L}\rho_0=0$, is derived in the Supplement, Section 2: 
\begin{equation}
\label{Sylvester}
\rho_0({\bf x}) \sim \exp\left( -\frac{1}{2} {\bf x}^\top {\bf \Sigma}^{-1} {\bf x} \right),
\end{equation}
where ${\bf x}\in\mathbb{R}^n$ and $\bf \Sigma$ is a solution of the Lyapunov equation, 
\begin{equation}
\label{lyapunov1}
{\bf A \Sigma} +  {\bf \Sigma A}^\top = -{\bf D}.
\end{equation}
For attractive dynamics, where all eigenvalues of $\bf A$ have negative real parts, Eq. \eqref{lyapunov1} has a unique positive-definite solution corresponding to the covariance matrix, ${\bf \Sigma} = \mathbb{E}\left[{\bf x}{\bf x}^\top\right]$. For repulsive or saddle-point dynamics, Eq. \eqref{lyapunov1} is a Sylvester equation, which has a unique solution if and only if $\bf A$ has no eigenvalues related by sign reversal \cite{bhatia1997and}. Thus, for a generic choice of $\bf A$, Eq. \eqref{lyapunov1} has a unique solution that is symmetric, real, and invertible. For repulsive directions, corresponding to eigenvalues of $\bf A$ with positive real parts, Eq. \eqref{Sylvester} characterizes how fast the density grows away from the fixed point~\cite{kwon2005structure}.



Under stationarity, $\mu({\bf x})=\nu({\bf x})=\rho_0({\bf x})$, the adjoint of $\mathcal{K}_\tau$ with respect to $\rho_0({\bf x})$ is given by
\begin{align}
\mathcal{K}^\dagger_\tau &=\text{diag}(\rho_0)^{-1}\exp{(\mathcal{L}\tau)}\text{diag}(\rho_0),
\end{align}
which is termed the reweighted Perron-Fronbenius operator and can also be interpreted as the SKO of the reverse-time dynamics. We then search for the eigenfunctions of the operators $\mathcal{F}_\tau$ and $\mathcal{B}_\tau$, Eq. \eqref{FB},
which are well defined on short time scales.

We find (see the Supplement, Section 3) a family of eigenfunctions of $\mathcal{F}_\tau,\mathcal{B}_\tau$ that can be expressed in terms of the linear projection of the state vector, $\mathbf x$:
\begin{align}
\label{eigen}
    \mathcal{F}_\tau ({\mathbf{v}_i}^\top{\bf x})=\lambda_i^2 ({\mathbf v_i}^\top{\bf x}),\quad \quad
    \mathcal{B}_\tau ({\mathbf{u}_i}^\top{\bf x})=\lambda_i^2 ({\mathbf u_i}^\top{\bf x}),
\end{align}
where $\mathbf u_i$ and $\mathbf v_i$ satisfy the following matrix eigenvector equations (see the Supplement, Section 3): 
\begin{align}
\label{sigma}
e^{{\bf A}^\top\tau}\Sigma^{-1}e^{{\bf A}\tau}\Sigma \mathbf v_i = \lambda_i^2 \mathbf v_i,\quad \quad
     \Sigma^{-1}e^{{\bf A}\tau}\Sigma e^{{\bf A}^\top\tau} \mathbf u_i = \lambda_i^2 \mathbf u_i.
\end{align}
Eigenvalues of transfer operators for expanding maps could be greater than one \cite{bandtlow2017spectral}. Since ideal coherent sets with zero leakage correspond to a unity eigenvalue, in this paper a subdominant eigenfunction refers to the non-trivial eigenfunction associated with an eigenvalue closest to one. That is, the subdominant eigenfunction for attractive (repulsive) dynamics is associated with the second largest (smallest) eigenvalue. We proved that the subdominant eigenfunctions belong to the linear family Eq. \eqref{eigen} for attractive and repulsive dynamics. Other eigenfunctions of $\mathcal{F}_\tau,\mathcal{B}_\tau$ are either a constant, associated with the eigenvalue one, that does not change sign, or represent higher-order polynomials in $\bf x$, associated with less dominant eigenvalues (see the Supplement, Section 3, which, however, does not consider saddle points). Therefore, to identify the least dispersive coherent sets, it is sufficient to focus on the eigenfunctions in  Eqs. \eqref{eigen}. The pair of minimally leaking coherent sets is given by the solution of Eqs. \eqref{sigma} with $\lambda_i^2$ closest to unity. To obtain the coherent set membership indices, we recover indicator functions, Eq. \eqref{eq:quotient}, by taking the sign of these singular functions. 

Eigenvectors ${\bf v}_i$ and ${\bf u}_i$ play a complementary role in prediction and retrospection. As the operator $\mathcal{F}_\tau$ propagates observables forward then backward in time, neurons projecting inputs onto ${\bf v}_i$ are predictive (Fig. 1b). As the operator $\mathcal{B}_\tau$ propagates the observable backward then forward in time, neurons projecting inputs onto ${\bf u}_i$ are retrospective (Fig. 1c). In the case of 2D saddle-point OU (Fig. 1b,c), when the shared past and future refer to other fixed points, the coherent sets of interest are the expanding coherent set for prediction (${\bf v}_1$) and the contracting coherent set for retrospection (${\bf u}_2$). In Section 4 of the Supplement, we analytically prove that, as the forecast horizon goes to infinity, singular vectors ${\bf v}_1$ and ${\bf u}_2$ converge to the top and bottom left eigenvector of $\bf A$ which are orthogonal to the stable and unstable invariant subspaces, respectively.

So far we considered fully observable OU processes. However, in reality, the OU process is often only partially observed for example via a linear projection of the state onto a scalar variable, 
\begin{align}
    y(t)= C {\bf x}(t).
\end{align}
Neurally, such projection can be an input to a single-synapse neuron or a summed total synaptic current. Assuming observability, the state can be represented by an $n$-dimensional lag vector \cite{katayama2005subspace}:
\begin{align}
    \hat{\bf x}(t)=[y(t), y(t-1), \dots, y(t-n+1)]^\top,
\end{align}
evolving via a companion matrix, equivalent to the following auto-regressive model:
\begin{align}
    y(t+1)=a_1y(t)+a_2y(t-1)+\cdots+a_ny(t-n+1)+\xi(t). 
\end{align}

\section{Coherent sets from data} \label{sec:datadriven}

The infinite-dimensional stochastic Koopman operator (SKO) can be approximated by a Galerkin projection onto a finite set of basis functions (or features), $\set{\phi_i(\mathbf{x})}_{i=1}^d$~\cite{williams2015,klus2024dynamical}. In the data-driven formulation,
the SKO matrix representation under this basis is~\cite{koltai2018,klus2024dynamical} (see Supplement, Section 5):
\begin{align}\label{eq:galerkin}
  {\bf K}_\tau &= {\bf \Sigma}_0^{-1}{\bf\Sigma}_\tau,
\end{align} 
where ${\bf \Sigma}_0$ and ${\bf \Sigma}_\tau$ are the covariance matrices: 
\vspace{-2mm}
\begin{equation}
\vspace{-2mm}
\label{phi}
    [{\bf \Sigma}_0]_{ij} \approx \frac{1}{S}\sum_{s=1}^S \phi_i({\bf X}(t_s)) \phi_j({\bf X}(t_s)),\quad\quad
    [{\bf \Sigma}_\tau]_{ij} \approx \frac{1}{S}\sum_{s=1}^S \phi_i({\bf X}(t_s)) \phi_j({\bf X}(t_s+\tau)),
\end{equation}
and describe the correlation between measurement functions $i$ and $j$ measured instantaneously, or with delay $\tau$, respectively. 

Such matrices are well defined if the process is stationary which may not be true for neurons. A common workaround is to assume local stationarity, that is, the distribution of inputs changes slowly compared to the timescale of the dynamics. In this case, the neuron can continuously update the estimate of the dynamics using exponential forgetting~\cite{widrow1985adaptive}. Our method can accommodate this type of non-stationarity by computing Eqs. \eqref{phi} locally in time, using a temporal filter that discounts older observations. This allows the estimated covariances to track gradual changes in the underlying distribution without assuming global stationarity.

As a direct consequence of Eq. \eqref{eq:innerprod}, the adjoint of ${\bf K}_{\tau}$ is ${\bf K}_\tau^\dagger={\bf \Sigma}_0^{-1}{\bf K}_\tau^\top {\bf\Sigma}_0$. The forward and backward operators have matrix representations~\cite{koltai2018,klus2024dynamical}:
\begin{align}
\label{cca}
 {\bf F}_\tau &={\bf K}_\tau{\bf\Sigma}_0^{-1}{\bf K}_\tau^\top {\bf\Sigma}_0 ={\bf\Sigma}_0^{-1}{\bf\Sigma}_\tau{\bf\Sigma}_0^{-1}{\bf\Sigma}_{-\tau}, & {\bf B}_\tau& ={\bf\Sigma}_0^{-1}{\bf K}_\tau^\top {\bf\Sigma}_0{\bf K}_\tau={\bf\Sigma}_0^{-1}{\bf\Sigma}_{-\tau}{\bf\Sigma}_0^{-1}{\bf\Sigma}_\tau,
\end{align}
which reduce to Eq. \eqref{sigma} for attractive OU processes (see the Supplement, Section 6).
The singular functions of the SKO within the Galerkin projection are given by 
\begin{align}\label{eq:ddsingularvectors}
    v_i({\bf x}) &= {\vec{ v}}_i \cdot  \vec{\phi}({\bf x}),\ & u_i({\bf x}) &= {\vec{ u}}_i \cdot  \vec{\phi}({\bf x}),
\end{align}
where ${\vec{ v}}_i$ and ${\vec{ u}}_i$ are eigenvectors of the matrices ${\bf F}_\tau$ and ${\bf B}_\tau$, respectively. Notice that ${\vec{ v}}_i$ and ${\vec{ u}}_i$ are precisely the solution of past-future CCA.

By representing these eigenvectors in synaptic weights and temporal filters, a neuron may compute a leading singular function of the SKO restricted to the linear span of its input features. By indicating the sign of the singular function, the neuron then computes the membership index for each input. For the sake of simplicity, below we only consider the temporal component of the singular vector.
Even if the neuron has access to only a single synapse, it can retain a history of that synapse’s activity over $d$ consecutive time steps. This construction---known as delay embedding or delay coordinates \cite{arbabi2017}---yields also a Galerkin approximation. We take ${\bf x}\in\mathbb{R}^d$ to be a lag vector of the input current to the neuron over the previous $d$ intervals of time, such that ${\bf X}(t) = (I(t), I(t-\Delta t), \dots, I(t-(d-1)\Delta t))$, where $I(t)$ is the observed current at time $t$. In this regime, the basis functions are given by
\begin{equation}
    \vspace{-1mm}
    \phi_i({\bf x}) = x_i,
    \label{lag basis}
\end{equation}
i.e.  $\vec{\phi}({\bf x})$ is the identity function. Singular vectors act as temporal filters over the delayed signal,
to \vspace{-2mm}
\begin{align}\label{eq:temporalweights}
r_{pre}^{\pm, i}(t) &= H\left(\pm v_i({\bf X}(t))\right) = H\left(\pm\sum_{k=0}^{d-1} [\vec{ v}_i]_k I(t-k\Delta t)\right)\nonumber,\\
r_{ret}^{\pm,i}(t) &= H\left(\pm u_i({\bf X}(t))\right)= H\left(\pm\sum_{k=0}^{d-1} [\vec{ u}_i]_k I(t-k\Delta t)\right),
\end{align}
where $\vec{ u}_i$ and $\vec{ v}_i$ are singular vectors (cf. Eq. \eqref{eq:ddsingularvectors}), and the Heaviside function, $H(x) = 
1$ for $x \ge 0$, and $H(x)=0$ for $x < 0$.

Therefore, sensory streams can be clustered into coherent sets via data-driven Galerkin approximations of the SKO. By learning the singular vectors from the features represented by the activity of the upstream neurons and encoding them in the synaptic weight and temporal filters, biological neurons can cluster inputs using integrate-and-fire dynamics.

\section{Coherent set clustering perspective on neurophysiology}

Here we apply the coherent set clustering framework to biological neurons focusing on three observations: temporal receptive fields, neuronal rectification and predictive/retrospective properties. We restrict our consideration to the processing of a scalar time series viewed as a one-dimensional projection of a multidimensional state-space dynamics, Eq. \eqref{eq:temporalweights}. Such a scalar time series could represent an input to a single-input neuron or the total synaptic current into a multi-input neuron. 
Here, inspired by our result for OU processes that the relevant singular functions lie in the span of a lag vector basis, we choose Eq. \eqref{lag basis} even in the data-driven, non-linear setting. Even such simple basis choice produces tangible results. 

\subsection{Temporal receptive fields}

We consider early visual processing where a natural stimulus can be generated by emulating the movement of the retinal image due to self-motion or saccades by scanning a natural image or its model. In turn, natural images are commonly modeled by the "dead leaves" model partitioning the space into patches of different but uniform luminance with sharp transitions between them \cite{jeulin1996dead}. Therefore, the resulting input time series is a set of plateaus at different levels with sharp transitions between them, Fig. 2a. Such scalar time series models input to post-photoreceptor neurons for invertebrates such as {\it Drosophila}, where initial processing is segregated between adjacent "pixels" or total current in a vertebrate bipolar cell. 

We view this scalar time series as a linear projection of a high-dimensional dynamical system. Thus, we perform past-future CCA on the lag-vectors formed from the scalar time series. The canonical correlations reveal a spectral gap following the top two. The lag-vector space is partitioned into a pair of coherent sets in a canonical direction. The top right singular vector amounts to a low-pass filter (Fig. 2b) similar to sustained bipolar cells of the vertebrate retina or L3 neuron in {\it Drosophila}~\cite{ketkar2022first}. The second canonical direction acts as a smoothed temporal derivative (Fig. 2b) in general agreement with experimentally reported filters of transient bipolars in vertebrates, Fig. 2c or L1 and L2 cells in {\it Drosophila} \cite{srinivasan1982predictive, ketkar2022first}. Because the stimulus is symmetric with respect to time-reversal, the complementary left singular vectors are obtained by simply inverting the time axis (and inverting the sign for the second singular vector).

We further characterize the experimentally measured temporal responses by interpreting them from the perspective of predictive and retrospective coherent sets. To make this connection, we take advantage of the observation that a practical predictive filter must be significantly aligned with an unstable eigendirection and a practical retrospective filter must be significantly aligned with a stable eigendirection, Fig. 1b,c. Moreover, for a two-dimensional (hyperbolic) saddle point, the linear temporal filter of the predictive neuron must be orthogonal to the attractive subspace, while the linear temporal filter of the retrospective neuron must be orthogonal to the repulsive subspace, Fig. 1b,c, see Section 3 and the Supplement, Section 4. Repulsive modes in the lag-vector space are expanding exponentials in real time and attractive modes are contracting exponentials, suggesting a simple way to interpret linear temporal receptive fields by computing the cosine similarity, $S$, with exponentials or, equivalently, computing normalized Laplace transforms, Fig. 2c,d. The experimentally measured temporal filter is orthogonal to a contracting exponential and aligned with the expanding exponential (Fig. 2c,d), indicating a predictive neuron. For details, see the Supplement, Section 7.
\begin{figure}[ht]
\vspace{-5mm}
  \centering
    \subfloat{\includegraphics[height=0.17\textwidth]{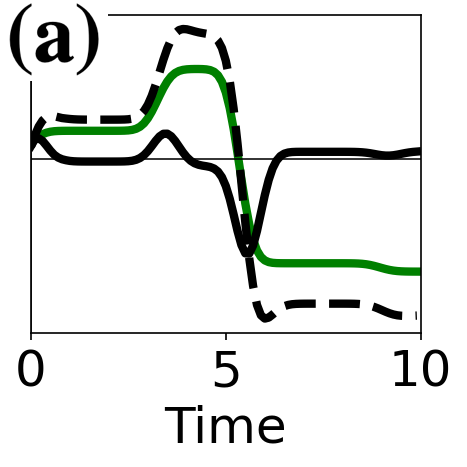}}\hspace{0mm}
    \subfloat{\includegraphics[height=0.17\textwidth]{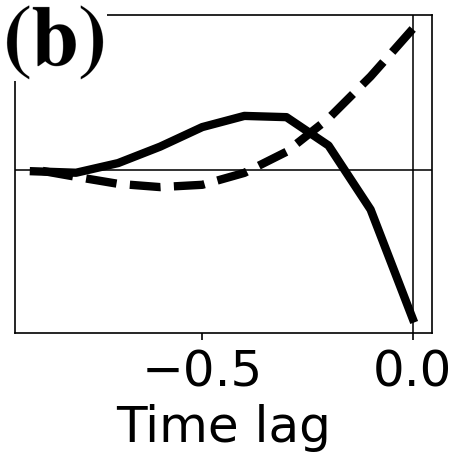}}\hspace{0mm}
    \subfloat{\includegraphics[height=0.17\textwidth]{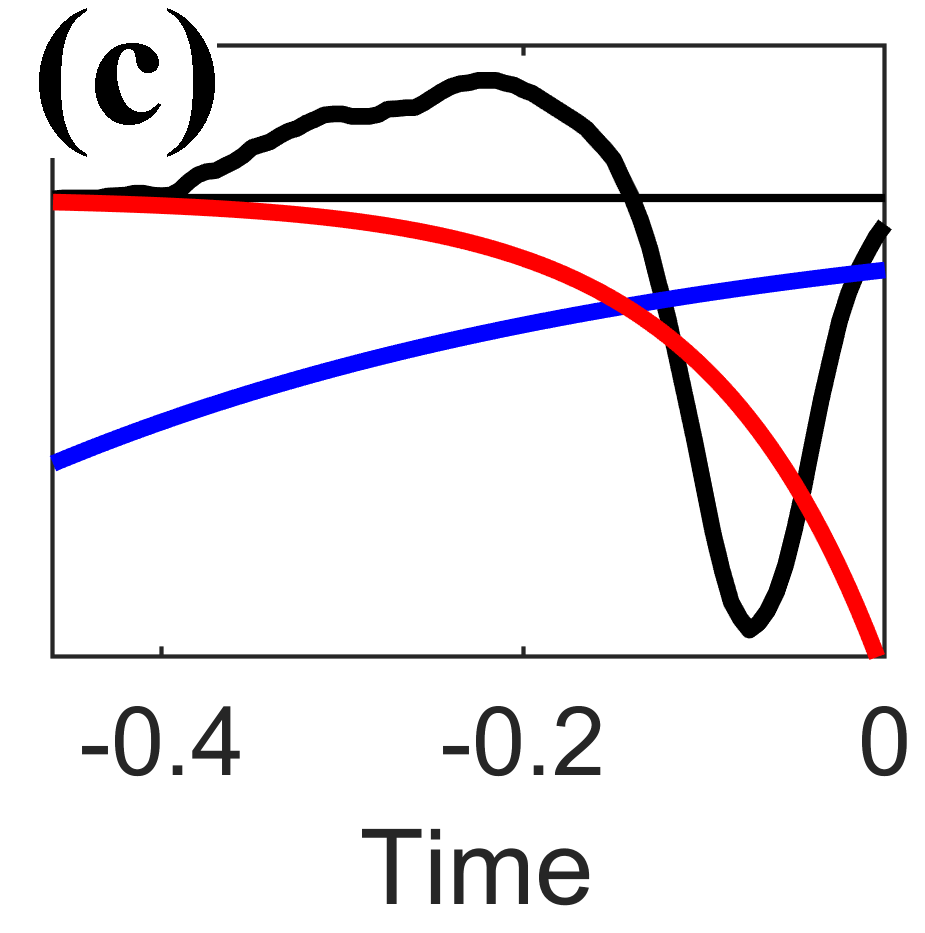}}\hspace{0mm}
    \subfloat{\includegraphics[height=0.17\textwidth]{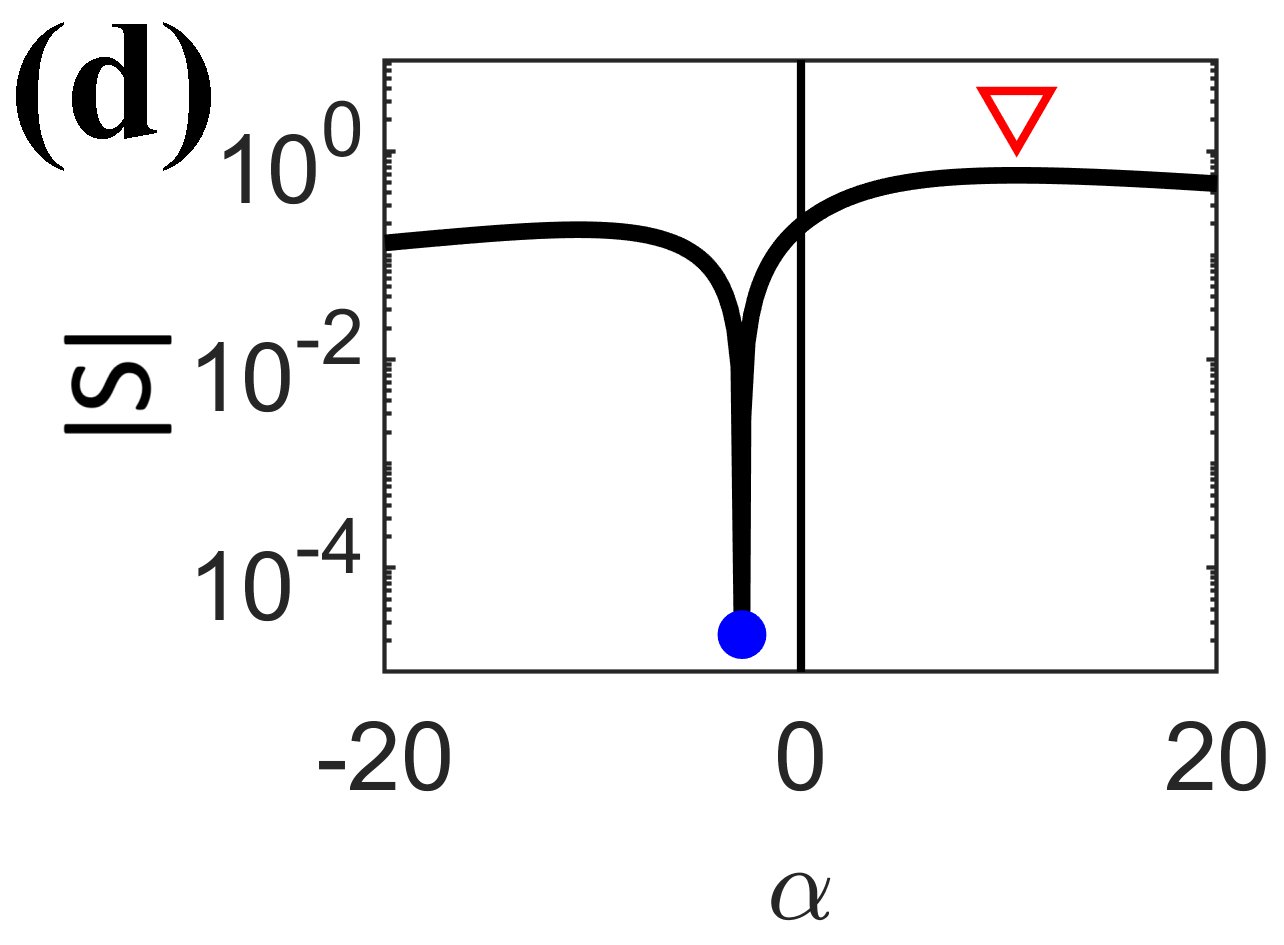}}
  \caption{Data-driven temporal filters learned on the scalar time series from the natural image model compared with that of a bipolar cell. 
  {\bf  (a)} Stimulus produced by the "dead leaves" model of natural images~\cite{jeulin1996dead} (green), stimulus filtered through first  (dashed black) and second (solid black) right singular vectors.
  {\bf (b)} First (dashed) and second (solid) right singular vectors from Eq. \eqref{cca} corresponding to predictive neurons. {\bf (c)} Black: experimentally measured flash response of the salamander retinal bipolar cell (D.B. Kastner \& S.A. Baccus, personal communication) on the inverted time axis approximates its linear filter. Compare this with with the second singular vector (solid) in {\bf b}. Note that the decay of the filter at zero time delay (absent in the theoretical result without additional constraints) is a consequence of causality and continuity of the filter implementation by a biological system. The maximally aligned (red) and the orthogonal (blue) exponentials are overlaid. {\bf (d)} Cosine similarities of the temporal filters with a battery of exponentials. Blue dot indicates orthogonality and red triangle indicates maximum alignment. Closer alignment with expanding exponentials and orthogonality to a contracting exponential indicates the neuron's predictive nature.
  }
  \vspace{-5mm}
\end{figure}
\subsection{Neuronal rectification exemplified by ON and OFF cells}

Our coherent set clustering perspective suggests that neurons determine membership indices by applying a Heaviside step-function to the SKO singular function (or its sign-inverse), Eq.\eqref{eq:temporalweights}. Such non-linearity could be naturally implemented by the spiking mechanism and correspond to the ON and OFF ganglion cells of the vertebrate retina, named this way because they respond to luminance increments and decrements, respectively. However, neurons post-synaptic to photoreceptors, bipolar cells in vertebrates and large monopolar cells in flies, are non-spiking. Neuronal activity is represented by continuously varying graded potentials that determine a non-negative synaptic vesicle release rate. Many such neurons thus rectify their input and are also classified as either ON or OFF cell. Such rectified response to luminance variation is graded and could be viewed as soft clustering by using not only the sign but also the magnitude of the singular function, Fig. 1, resulting in a rectified linear unit (ReLU)-like operation.  


\subsection{Predictive and retrospective properties of biological neurons}

In this Subsection, we analyze and review large datasets of experimentally measured temporal receptive fields from the perspective of predictive and retrospective coherent sets. 

\textbf{Tufted and mitral cells of the mammalian olfactory bulb.} In mammals, information from olfactory sensory neurons is relayed to the rest of the brain by two neuron classes: tufted cells (TCs) and mitral cells (MCs). We analyzed a dataset of the temporal receptive fields of 204 TCs and MCs recorded blindly from the rat olfactory bulb \cite{gupta2015olfactory} from the coherent set clustering perspective. As described above (5.1, Temporal Receptive Fields), we computed the cosine similarity of experimentally measured temporal filters with a battery of growing and decaying exponentials. In addition to identifying the orthogonal exponential, we also identified the sign of the exponent which has the highest cosine similarity with the temporal filter: positive exponents correspond to predictive neurons, and negative to retrospective. As a result, we found a mixture of predictive and retrospective properties, Fig. 3. Approximately half of the recorded cells had temporal filters orthogonal to decaying exponentials and, therefore, were likely predictive. About $4\%$ were orthogonal to growing exponentials and, hence, were likely retrospective. About one-third were orthogonal to both growing and decaying exponentials suggesting that they analyze a higher- than 2-dimensional dynamics and could be either retrospective or predictive. See the Supplement, Section 7 for details.

We speculate that TCs and MCs are mostly predictive and retrospective, respectively, based on the existing literature. First, they respond at different phases of the sniff cycle as monitored by air flow \cite{fukunaga2012two}: TCs are preferentially active during the exponential growth phase and silent during the exponential decay phase. Conversely, MCs are preferentially active during the exponential decay phase and silent during the exponential growth phase. Second, MC responses lag relative to TCs suggesting they play a different role \cite{fukunaga2012two,abraham2004maintaining,uchida2003speed,rinberg2006speed}, but see \cite{diaz2018inhalation, short2019temporal}.
Finally, TCs and MCs differ in their responses to synaptic inputs from the electrophysiologically stimulated olfactory sensory neurons~\cite{gire2012mitral}: TCs receive direct inputs and respond with an immediate biphasic profile typical of prediction, while MCs are activated indirectly with a delay characteristic of retrospection.
\begin{figure}[h!]
  \centering
  \includegraphics[width=0.8\textwidth]{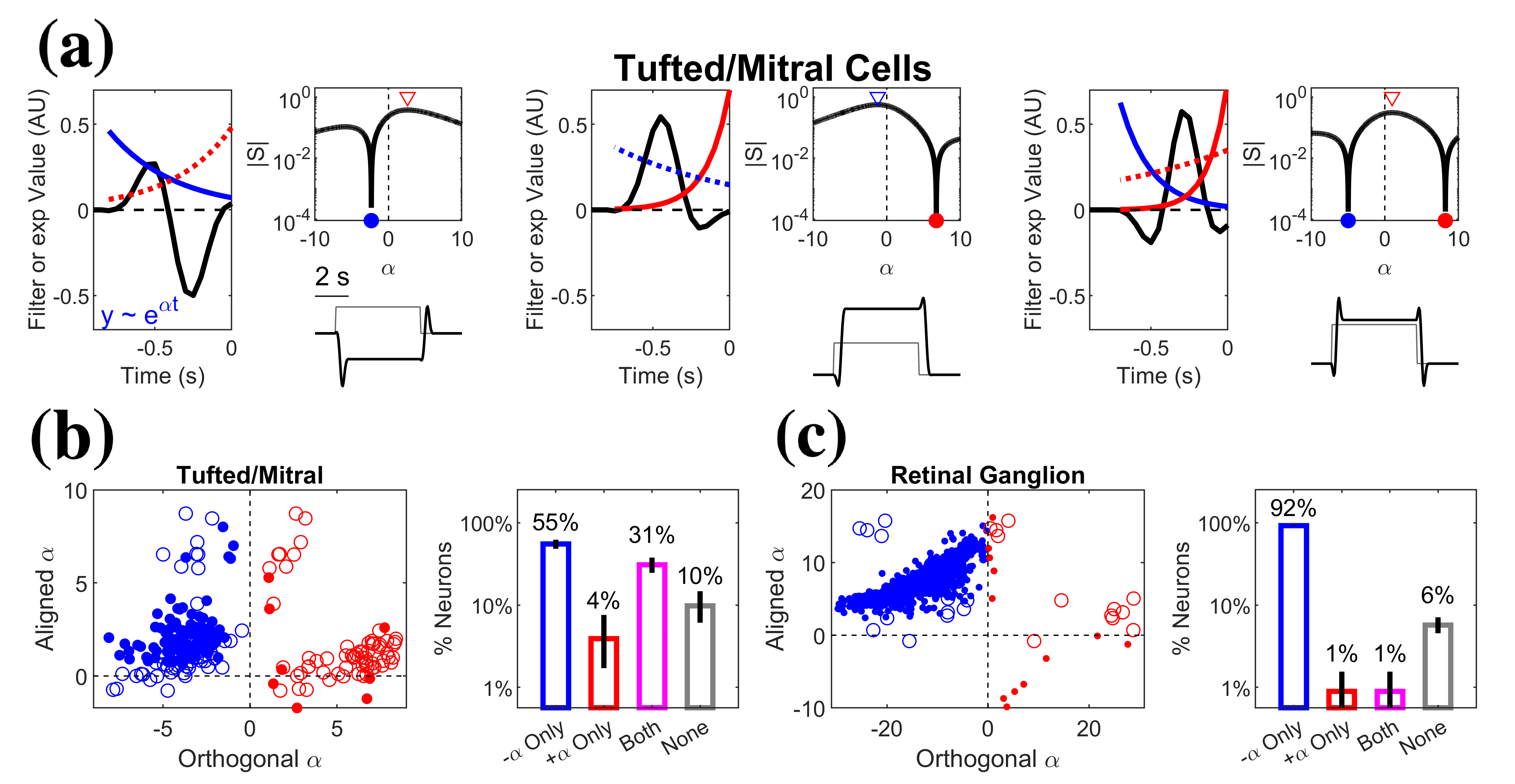}
  \vspace{-1mm}
  \caption{Temporal receptive fields of olfactory bulb and retinal neurons interpreted through the lens of coherent sets exhibit predictive and retrospective properties. {\bf (a)} Linear temporal filters (black) from three rat olfactory bulb tufted/mitral neurons~\cite{gupta2015olfactory}. Solid colored lines are orthogonal exponentials; dotted colored lines are aligned exponentials. Top insets show the cosine similarity, $S$, of the temporal filter with exponentials. Orthogonal and aligned values are marked by dots and triangles, respectively. Bottom insets show convolution of the filter with a step pulse. The first neuron is orthogonal to a decaying and aligned with expanding exponentials and interpreted as predictive. The second neuron is orthogonal to an expanding and aligned with decaying exponentials and interpreted as retrospective. The third neuron is orthogonal to both a decaying and expanding exponential but aligned to an expanding one. {\bf (b)} Left, for each tufted/mitral cell, the $\alpha$ value yielding an orthogonal exponential is plotted against the $\alpha$ that is maximally aligned. Filled dots indicate neurons orthogonal to a single $\alpha$. Open dots indicate neurons orthogonal to two $\alpha$ values. Right, The proportion of cells orthogonal to a negative exponential only, a positive exponential only, both, or neither is shown in the bar graph. 95\% confidence intervals of the mean obtained through a binomial fit are plotted as black lines. {\bf (c)} Same as {\bf b} but for retinal ganglion cells (RGCs). See Fig. S1 for example temporal receptive fields. Note that very few RGCs are orthogonal to expanding exponentials.}
  \vspace{-5mm}
\end{figure}

\textbf{Retinal ganglion cells (RGCs).} To explore whether the distributions of predictive and retrospective neurons vary across sensory modalities, we analyzed a large dataset of temporal receptive fields of 1345 RGCs from dissociated vertebrate retina \cite{KastnerBaccus2011}, Figs. S1, 3C. In vertebrates, signals from differentially stimulated photoreceptors are combined in the retina and, therefore, an RGC's receptive field has a spatial component. Here, we only focus on the temporal component obtained by a rank-1 decomposition of the spatio-temporal receptive field. The majority of RGCs' receptive fields were orthogonal to decaying exponentials and aligned with expanding exponentials, hence  predictive, though $1\%$ of cells were retrospective, Fig. 3C. See the Supplement, Section 7 for details.

\textbf{Lagged and non-lagged cells in the Lateral Geniculate Nucleus.} RGCs project to the Lateral Geniculate Nucleus (LGN) \cite{cai1997spatiotemporal}, in which two classes of relay cells have been described in cat \cite{mastronarde1987two, wolfe1998temporal, saul1990spatial}, monkey \cite{saul2008lagged}, and mouse \cite{piscopo2013diverse}: one with the temporal receptive fields of predictive cells and one with retrospective properties, Fig. S2. Because the firing of the latter class of neurons substantially lags a step stimulus onset, they were termed "lagged cells". Comparing the response profiles of our dataset of tufted/mitral cells to a step stimulus provides insight into how lagged cells may manifest, Fig. 3A. Note that the middle, retrospective, neuron is initially inhibited by the positive step response, but overcomes this inhibition due to the delayed positive filter lobe. Thus, we identify non-lagged and lagged neurons with predictive and retrospective units, respectively.
\section{Discussion}
We propose that neurons act as \emph{coherent–set detectors}: each unit computes a linear projection of its input onto a singular vector of a whitened finite–time transition matrix and then takes the sign or rectifies its positive or negative part. We derive this mechanism analytically for saddle–point Ornstein–Uhlenbeck (OU) dynamics and suggest that it extends to nonlinear systems via Galerkin projection on a chosen feature basis. In practice, a leading singular function can be learned directly from data, providing a path toward biologically plausible implementations (e.g., local weight updates). This perspective captures several physiological regularities, including rectification and selective sensitivity to temporally extended (finite–time) structure in stimuli as well as predictive and retrospective neuron types.

\textbf{Open issues:}\\
\textit{(i) Unstable modes and reference measures.} Section~3 characterizes SKO singular functions for the OU model, but linking them to the spectral clustering framework of Section~2 under \emph{instability} requires a reference measure that decays away from the saddle neighborhood. Practically, this suggests finite–time localization (windowing), exponential discounting, or committor/Doob–transformed weights so that the inner product prioritizes trajectories that remain near the saddle. Interpreting the OU model as a local description then would render the spectral objects well posed on that neighborhood.\\
\textit{(ii) Data–driven estimation along unstable directions.} The Galerkin procedure of Section~4 matches the OU analysis for stable modes; extending it to unstable modes likely requires modified forward–backward compositions over short horizons so that finite–time expansion (singular values (>1)) is preserved rather than suppressed by whitening. Establishing a rigorous equivalence in this regime remains open.\\
\textit{(iii) Beyond 2D saddles.} The correspondence between finite–time singular vectors and eigenvectors shown for 2D saddles may not extend to higher dimensions, especially with multiple unstable directions. Systematic analysis in higher–dimensional settings is needed.\\
\textit{(iv) Circuit feedback loops.} While many early sensory pathways are predominantly feedforward and fall within our framework, feedback is ubiquitous in deeper circuits. Incorporating explicit closed–loop interactions—for example, via the controller perspective of~\cite{moore2024neuron}—could unify coherent–set detection with action selection and behavior generation.\\
\textit{(v) Intra–neuronal feedback.} In our current account, rectification follows linear projection and does not influence learning of synaptic weights or temporal filters. In reality, synaptic plasticity may depend on rectified output via dendritic backpropagation. Incorporating such output into projection learning could provide the needed windowing for local learning mechanisms.

\textbf{Societal impact.}
By linking biological neural computation to finite-time dynamical structure, this work advances our understanding of brain function and may inform approaches to mental health and neurological disorders. It could also inspire biologically grounded, self-supervised architectures for artificial neural networks operating in dynamical settings.

\textbf{Acknowledgments.} We are grateful to  T. Ahamed, R. Cosentino, A. Genkin, B. Gutkin, J. M. Heninger, P. Kidd, A. Koulakov, I. Nemenman, S.E. Palmer, S. Qin, D. Rinberg, E. Schneidman, D. Schwab, E. Simoncelli, S. Solla, P. Thomas and, especially,  P. Cvitanovic for insightful discussions, to P. Gupta, D.F. Albeanu, D.B. Kastner, and S.A. Baccus for sharing their datasets. This research was supported in part by the National Science Foundation (NSF) EFRI BRAID: Scalable-Learning Neuromorphics, award No. 2318152. Some work was initiated and performed at the the Kavli Institute for Theoretical Physics (KITP) supported by grant NSF PHY-2309135 and at the Aspen Center for Physics, which is supported by NSF grant PHY-2210452.



\bibliographystyle{unsrt}
\bibliography{references}

@article{srinivasan1982predictive,
  title={Predictive coding: a fresh view of inhibition in the retina},
  author={Srinivasan, Mandyam Veerambudi and Laughlin, Simon Barry and Dubs, Andreas},
  journal={Proceedings of the Royal Society of London. Series B. Biological Sciences},
  volume={216},
  number={1205},
  pages={427--459},
  year={1982},
  publisher={The Royal Society London}
}

@inproceedings{cvitanovic2012knowing,
  title={Knowing when to stop: How noise frees us from determinism},
  author={Cvitanovi{\'c}, Predrag and Lippolis, Domenico},
  booktitle={AIP Conference Proceedings},
  volume={1468},
  number={1},
  pages={82--126},
  year={2012},
  organization={American Institute of Physics}
}

@article{kwon2005structure,
  title={Structure of stochastic dynamics near fixed points},
  author={Kwon, Chulan and Ao, Ping and Thouless, David J},
  journal={Proceedings of the National Academy of Sciences},
  volume={102},
  number={37},
  pages={13029--13033},
  year={2005},
  publisher={National Academy of Sciences}
}

@article{chalk2018toward,
  title={Toward a unified theory of efficient, predictive, and sparse coding},
  author={Chalk, Matthew and Marre, Olivier and Tka{\v{c}}ik, Ga{\v{s}}per},
  journal={Proceedings of the National Academy of Sciences},
  volume={115},
  number={1},
  pages={186--191},
  year={2018},
  publisher={National Acad Sciences}
}

@article{singer2018sensory,
  title={Sensory cortex is optimized for prediction of future input},
  author={Singer, Yosef and Teramoto, Yayoi and Willmore, Ben DB and Schnupp, Jan WH and King, Andrew J and Harper, Nicol S},
  journal={Elife},
  volume={7},
  pages={e31557},
  year={2018},
  publisher={eLife Sciences Publications Limited}
}

@article{wu2020variational,
  title={Variational approach for learning Markov processes from time series data},
  author={Wu, Hao and No{\'e}, Frank},
  journal={Journal of Nonlinear Science},
  volume={30},
  number={1},
  pages={23--66},
  year={2020},
  publisher={Springer}
}

@article{dellnitz1999approximation,
  title={On the approximation of complicated dynamical behavior},
  author={Dellnitz, Michael and Junge, Oliver},
  journal={SIAM Journal on Numerical Analysis},
  volume={36},
  number={2},
  pages={491--515},
  year={1999},
  publisher={SIAM}
}

@article{rao1999predictive,
  title={Predictive coding in the visual cortex: a functional interpretation of some extra-classical receptive-field effects},
  author={Rao, Rajesh PN and Ballard, Dana H},
  journal={Nature neuroscience},
  volume={2},
  number={1},
  pages={79--87},
  year={1999},
  publisher={Nature Publishing Group}
}

@article{froyland2013analytic,
  title={An analytic framework for identifying finite-time coherent sets in time-dependent dynamical systems},
  author={Froyland, Gary},
  journal={Physica D: Nonlinear Phenomena},
  volume={250},
  pages={1--19},
  year={2013},
  publisher={Elsevier}
}

@article{tishby2000information,
  title={The information bottleneck method},
  author={Tishby, Naftali and Pereira, Fernando C and Bialek, William},
  journal={arXiv preprint physics/0004057},
  year={2000}
}

@inproceedings{bialek2006efficient,
  title={Efficient representation as a design principle for neural coding and computation},
  author={Bialek, William and Van Steveninck, Rob R De Ruyter and Tishby, Naftali},
  booktitle={2006 IEEE international symposium on information theory},
  pages={659--663},
  year={2006},
  organization={IEEE}
}

@article{palmer2015predictive,
  title={Predictive information in a sensory population},
  author={Palmer, Stephanie E and Marre, Olivier and Berry, Michael J and Bialek, William},
  journal={Proceedings of the National Academy of Sciences},
  volume={112},
  number={22},
  pages={6908--6913},
  year={2015},
  publisher={National Acad Sciences}
}

@article{attneave1954some,
  title={Some informational aspects of visual perception.},
  author={Attneave, Fred},
  journal={Psychological review},
  volume={61},
  number={3},
  pages={183},
  year={1954},
  publisher={American Psychological Association}
}

@article{barlow1961possible,
  title={Possible principles underlying the transformation of sensory messages},
  author={Barlow, Horace B and others},
  journal={Sensory communication},
  volume={1},
  number={01},
  pages={217--233},
  year={1961}
}

@article{lipshutz2021biologically,
  title={A biologically plausible neural network for multichannel canonical correlation analysis},
  author={Lipshutz, David and Bahroun, Yanis and Golkar, Siavash and Sengupta, Anirvan M and Chklovskii, Dmitri B and many, others},
  journal={Neural Computation},
  volume={33},
  number={9},
  pages={2309--2352},
  year={2021},
  publisher={MIT Press One Rogers Street, Cambridge, MA 02142-1209, USA journals-info~…}
}

@article{fukunaga2012two,
  title={Two distinct channels of olfactory bulb output},
  author={Fukunaga, Izumi and Berning, Manuel and Kollo, Mihaly and Schmaltz, Anja and Schaefer, Andreas T},
  journal={Neuron},
  volume={75},
  number={2},
  pages={320--329},
  year={2012},
  publisher={Elsevier}
}

@article{lipshutz2020biologically,
  title={A biologically plausible neural network for slow feature analysis},
  author={Lipshutz, David and Windolf, Charles and Golkar, Siavash and Chklovskii, Dmitri B},
  journal={Advances in Neural Information Processing Systems},
  volume={33},
  pages={14986--14996},
  year={2020}
}

@article{cai1997spatiotemporal,
  title={Spatiotemporal receptive field organization in the lateral geniculate nucleus of cats and kittens},
  author={Cai, Daqing and Deangelis, Gregory C and Freeman, Ralph D},
  journal={Journal of neurophysiology},
  volume={78},
  number={2},
  pages={1045--1061},
  year={1997},
  publisher={American Physiological Society Bethesda, MD}
}

@article{mastronarde1987two,
  title={Two classes of single-input X-cells in cat lateral geniculate nucleus. I. Receptive-field properties and classification of cells},
  author={Mastronarde, David N},
  journal={Journal of neurophysiology},
  volume={57},
  number={2},
  pages={357--380},
  year={1987},
  publisher={American Physiological Society Bethesda, MD}
}

@article{pavliotis2014stochastic,
  title={Stochastic processes and applications},
  author={Pavliotis, Grigorios A},
  journal={Texts in applied mathematics},
  volume={60},
  year={2014},
  publisher={Springer}
}

@article{moore2024neuron,
  title={The neuron as a direct data-driven controller},
  author={Moore, Jason J and Genkin, Alexander and Tournoy, Magnus and Pughe-Sanford, Joshua L and de Ruyter van Steveninck, Rob R and Chklovskii, Dmitri B},
  journal={Proceedings of the National Academy of Sciences},
  volume={121},
  number={27},
  pages={e2311893121},
  year={2024},
  publisher={National Academy of Sciences}
}

@article{piscopo2013diverse,
  title={Diverse visual features encoded in mouse lateral geniculate nucleus},
  author={Piscopo, Denise M and El-Danaf, Rana N and Huberman, Andrew D and Niell, Cristopher M},
  journal={Journal of Neuroscience},
  volume={33},
  number={11},
  pages={4642--4656},
  year={2013},
  publisher={Society for Neuroscience}
}

@article{dong1995temporal,
  title={Temporal decorrelation: a theory of lagged and nonlagged responses in the lateral geniculate nucleus},
  author={Dong, Dawei W and Atick, Joseph J},
  journal={Network: Computation in Neural Systems},
  volume={6},
  number={2},
  pages={159--178},
  year={1995},
  publisher={Taylor \& Francis}
}

@misc{ enwiki:1286449686,
    author = "{Wikipedia contributors}",
    title = "Ornstein–Uhlenbeck process --- {Wikipedia}{,} The Free Encyclopedia",
    year = "2025",
    url = "https://en.wikipedia.org/w/index.php?title=Ornstein%E2%80%93Uhlenbeck_process&oldid=1286449686",
    note = "[Online; accessed 15-May-2025]"
  }

@article{gjorgjieva2014benefits,
  title={Benefits of pathway splitting in sensory coding},
  author={Gjorgjieva, Julijana and Sompolinsky, Haim and Meister, Markus},
  journal={Journal of Neuroscience},
  volume={34},
  number={36},
  pages={12127--12144},
  year={2014},
  publisher={Society for Neuroscience}
}

@incollection{tuckwell2012stochastic,
  title={Stochastic partial differential equations in neurobiology: linear and nonlinear models for spiking neurons},
  author={Tuckwell, Henry C},
  booktitle={Stochastic biomathematical models: with applications to neuronal modeling},
  pages={149--173},
  year={2012},
  publisher={Springer}
}

@article{wolfe1998temporal,
  title={Temporal diversity in the lateral geniculate nucleus of cat},
  author={Wolfe, Jonathan and Palmer, Larry A},
  journal={Visual neuroscience},
  volume={15},
  number={4},
  pages={653--675},
  year={1998},
  publisher={Cambridge University Press}
}

@article{saul1990spatial,
  title={Spatial and temporal response properties of lagged and nonlagged cells in cat lateral geniculate nucleus},
  author={Saul, Alan B and Humphrey, AL},
  journal={Journal of neurophysiology},
  volume={64},
  number={1},
  pages={206--224},
  year={1990}
}

@article{saul2008lagged,
  title={Lagged cells in alert monkey lateral geniculate nucleus},
  author={Saul, Alan B},
  journal={Visual neuroscience},
  volume={25},
  number={5-6},
  pages={647--659},
  year={2008},
  publisher={Cambridge University Press}
}

@conference {pehlevan2017clustering,
	title = {A clustering neural network model of insect olfaction},
	booktitle = {2017 51st Asilomar Conference on Signals, Systems, and Computers},
	year = {2017},
	pages = {593{\textendash}600},
	publisher = {IEEE},
	organization = {IEEE},
	author = {Cengiz Pehlevan and Genkin, Alexander and Chklovskii, Dmitri B}
}

@article{uchida2003speed,
  title={Speed and accuracy of olfactory discrimination in the rat},
  author={Uchida, Naoshige and Mainen, Zachary F},
  journal={Nature neuroscience},
  volume={6},
  number={11},
  pages={1224--1229},
  year={2003},
  publisher={Nature Publishing Group US New York}
}

@article{abraham2004maintaining,
  title={Maintaining accuracy at the expense of speed: stimulus similarity defines odor discrimination time in mice},
  author={Abraham, Nixon M and Spors, Hartwig and Carleton, Alan and Margrie, Troy W and Kuner, Thomas and Schaefer, Andreas T},
  journal={neuron},
  volume={44},
  number={5},
  pages={865--876},
  year={2004},
  publisher={Elsevier}
}

@article{rinberg2006speed,
  title={Speed-accuracy tradeoff in olfaction},
  author={Rinberg, Dmitry and Koulakov, Alexei and Gelperin, Alan},
  journal={Neuron},
  volume={51},
  number={3},
  pages={351--358},
  year={2006},
  publisher={Elsevier}
}

@article{gire2012mitral,
  title={Mitral cells in the olfactory bulb are mainly excited through a multistep signaling path},
  author={Gire, David H and Franks, Kevin M and Zak, Joseph D and Tanaka, Kenji F and Whitesell, Jennifer D and Mulligan, Abigail A and Hen, Ren{\'e} and Schoppa, Nathan E},
  journal={Journal of Neuroscience},
  volume={32},
  number={9},
  pages={2964--2975},
  year={2012},
  publisher={Society for Neuroscience}
}

@article{gupta2015olfactory,
  title={Olfactory bulb coding of odors, mixtures and sniffs is a linear sum of odor time profiles},
  author={Gupta, Priyanka and Albeanu, Dinu F and Bhalla, Upinder S},
  journal={Nature neuroscience},
  volume={18},
  number={2},
  pages={272--281},
  year={2015},
  publisher={Nature Publishing Group US New York}
}

@article{simoncelli2001natural,
  title={Natural image statistics and neural representation},
  author={Simoncelli, Eero P and Olshausen, Bruno A},
  journal={Annual review of neuroscience},
  volume={24},
  number={1},
  pages={1193--1216},
  year={2001},
  publisher={Annual Reviews 4139 El Camino Way, PO Box 10139, Palo Alto, CA 94303-0139, USA}
}

@inproceedings{jeulin1996dead,
  title={Dead leaves models: from space tesselation to random functions},
  author={Jeulin, D},
  booktitle={Symposium on the Advances in the Theory and Applications of Random Sets},
  pages={137--156},
  year={1996},
  organization={World Scientific}
}

@article{gregor2012lattice,
  title={A lattice filter model of the visual pathway},
  author={Gregor, Karol and Chklovskii, Dmitri},
  journal={Advances in Neural Information Processing Systems},
  volume={25},
  year={2012}
}

@article{klus2024dynamical,
  title={Dynamical systems and complex networks: A Koopman operator perspective},
  author={Klus, Stefan and Conrad, Nata{\v{s}}a Djurdjevac},
  journal={Journal of Physics: Complexity},
  volume={5},
  number={4},
  pages={041001},
  year={2024},
  publisher={IOP Publishing}
}

@inproceedings{choi2024koopman,
  title={Koopman Autoencoder Via Singular Value Decomposition for Data-Driven Long-Term Prediction},
  author={Choi, Jinho and Krishnan, Sivaram and Park, Jihong},
  booktitle={2024 IEEE 34th International Workshop on Machine Learning for Signal Processing (MLSP)},
  pages={1--6},
  year={2024},
  organization={IEEE}
}

@article{golkar2024neuronal,
  title={Neuronal temporal filters as normal mode extractors},
  author={Golkar, Siavash and Berman, Jules and Lipshutz, David and Haret, Robert Mihai and Gollisch, Tim and Chklovskii, Dmitri B},
  journal={Physical Review Research},
  volume={6},
  number={1},
  pages={013111},
  year={2024},
  publisher={APS}
}

@article{heninger2015neighborhoods,
  title={Neighborhoods of periodic orbits and the stationary distribution of a noisy chaotic system},
  author={Heninger, Jeffrey M and Lippolis, Domenico and Cvitanovi{\'c}, Predrag},
  journal={Physical Review E},
  volume={92},
  number={6},
  pages={062922},
  year={2015},
  publisher={APS}
}

@article{ma2006bayesian,
  title={Bayesian inference with probabilistic population codes},
  author={Ma, Wei Ji and Beck, Jeffrey M and Latham, Peter E and Pouget, Alexandre},
  journal={Nature neuroscience},
  volume={9},
  number={11},
  pages={1432--1438},
  year={2006},
  publisher={Nature Publishing Group US New York}
}

@article{wiskott2002slow,
  title={Slow feature analysis: Unsupervised learning of invariances},
  author={Wiskott, Laurenz and Sejnowski, Terrence J},
  journal={Neural Computation},
  volume={14},
  number={4},
  pages={715--770},
  year={2002},
  publisher={MIT Press}
}

@article{bhatia1997and,
  title={How and why to solve the operator equation AX- XB= Y},
  author={Bhatia, Rajendra and Rosenthal, Peter},
  journal={Bulletin of the London Mathematical Society},
  volume={29},
  number={1},
  pages={1--21},
  year={1997},
  publisher={Cambridge University Press}
}

@article{DestexheOU_2001,
  title={Fluctuating synaptic conductances recreate in vivo-like activity in neocortical neurons},
  author={Destexhe, Alain and Rudolph, M and Fellous, J M and Sejnowski, T J},
  journal={Neuroscience},
  volume={107},
  number={1},
  pages={13-24},
  year={2001},
  publisher={Elsevier}
}

@article{ketkar2022first,
  title={First-order visual interneurons distribute distinct contrast and luminance information across ON and OFF pathways to achieve stable behavior},
  author={Ketkar, Madhura D and G{\"u}r, Burak and Molina-Obando, Sebastian and Ioannidou, Maria and Martelli, Carlotta and Silies, Marion},
  journal={Elife},
  volume={11},
  pages={e74937},
  year={2022},
  publisher={eLife Sciences Publications Limited}
}

@article{anderson1982reverse,
  title={Reverse-time diffusion equation models},
  author={Anderson, Brian DO},
  journal={Stochastic Processes and their Applications},
  volume={12},
  number={3},
  pages={313--326},
  year={1982},
  publisher={Elsevier}
}

@article{KastnerBaccus2011,
   title={ Coordinated dynamic encoding in the retina using opposing forms of plasticity},
   author={Kastner, David B and Baccus, Stephen A},
   journal={Nature Neuroscience},
   volume={14},
   number={10},
   pages={1317-1322},
   year={2011}
}

@article{williams2015,
  title = {A kernel-based method for data-driven koopman spectral analysis},
  volume = {2},
  DOI = {10.3934/jcd.2015005},
  number = {2},
  journal = {Journal of Computational Dynamics},
  publisher = {American Institute of Mathematical Sciences (AIMS)},
  author = {O.  Williams,  Matthew and W. Rowley,  Clarence and G.  Kevrekidis,  Ioannis},
  year = {2015},
  pages = {247–265}
}

@article{arbabi2017,
  title = {Ergodic Theory,  Dynamic Mode Decomposition,  and Computation of Spectral Properties of the Koopman Operator},
  volume = {16},
  ISSN = {1536-0040},
  url = {http://dx.doi.org/10.1137/17M1125236},
  DOI = {10.1137/17m1125236},
  number = {4},
  journal = {SIAM Journal on Applied Dynamical Systems},
  publisher = {Society for Industrial & Applied Mathematics (SIAM)},
  author = {Arbabi,  Hassan and Mezić,  Igor},
  year = {2017},
  month = jan,
  pages = {2096–2126}
}

@book{katayama2005subspace,
  title={Subspace methods for system identification},
  author={Katayama, Tohru},
  year={2005},
  publisher={Springer}
}

@book{widrow1985adaptive,
  title={Adaptive Signal Processing},
  author={Widrow, Bernard and Stearns, Samuel D},
  year={1985},
  publisher={Prentice-Hall},
  address={Englewood Cliffs, NJ}
}

@article{froyland2010transport,
  title={Transport in time-dependent dynamical systems: Finite-time coherent sets},
  author={Froyland, Gary and Santitissadeekorn, Naratip and Monahan, Adam},
  journal={Chaos: An Interdisciplinary Journal of Nonlinear Science},
  volume={20},
  number={4},
  year={2010},
  publisher={AIP Publishing}
}

@article{mardt2018vampnets,
  title={VAMPnets for deep learning of molecular kinetics},
  author={Mardt, Andreas and Pasquali, Luca and Wu, Hao and No{\'e}, Frank},
  journal={Nature communications},
  volume={9},
  number={1},
  pages={5},
  year={2018},
  publisher={Nature Publishing Group UK London}
}

@article{koltai2018,
  title = {Optimal Data-Driven Estimation of Generalized Markov State Models for Non-Equilibrium Dynamics},
  volume = {6},
  ISSN = {2079-3197},
  DOI = {10.3390/computation6010022},
  number = {1},
  journal = {Computation},
  publisher = {MDPI AG},
  author = {Koltai,  Péter and Wu,  Hao and Noé,  Frank and Sch\"{u}tte,  Christof},
  year = {2018},
  month = feb,
  pages = {22}
}

@article{chechik2003information,
  title={Information bottleneck for Gaussian variables},
  author={Chechik, Gal and Globerson, Amir and Tishby, Naftali and Weiss, Yair},
  journal={Advances in Neural Information Processing Systems},
  volume={16},
  year={2003}
}

@article{price2022efficient,
  title={Efficient temporal coding in the early visual system: existing evidence and future directions},
  author={Price, Byron H and Gavornik, Jeffrey P},
  journal={Frontiers in Computational Neuroscience},
  volume={16},
  pages={929348},
  year={2022},
  publisher={Frontiers Media SA}
}

@article{druckmann2012mechanistic,
  title={A mechanistic model of early sensory processing based on subtracting sparse representations},
  author={Druckmann, Shaul and Hu, Tao and Chklovskii, Dmitri},
  journal={Advances in Neural Information Processing Systems},
  volume={25},
  year={2012}
}

@article{diaz2018inhalation,
  title={Inhalation frequency controls reformatting of mitral/tufted cell odor representations in the olfactory bulb},
  author={D{\'\i}az-Quesada, Marta and Youngstrom, Isaac A and Tsuno, Yusuke and Hansen, Kyle R and Economo, Michael N and Wachowiak, Matt},
  journal={Journal of Neuroscience},
  volume={38},
  number={9},
  pages={2189--2206},
  year={2018},
  publisher={Society for Neuroscience}
}

@article{short2019temporal,
  title={Temporal dynamics of inhalation-linked activity across defined subpopulations of mouse olfactory bulb neurons imaged in vivo},
  author={Short, Shaina M and Wachowiak, Matt},
  journal={eneuro},
  volume={6},
  number={3},
  year={2019},
  publisher={Society for Neuroscience}
}

@inproceedings{bandtlow2017spectral,
  title={Spectral structure of transfer operators for expanding circle maps},
  author={Bandtlow, Oscar F and Just, Wolfram and Slipantschuk, Julia},
  booktitle={Annales de l'Institut Henri Poincar{\'e} C, Analyse non lin{\'e}aire},
  volume={34},
  number={1},
  pages={31--43},
  year={2017},
  organization={Elsevier}
}

\end{document}


\author{%
  \hspace{-1mm}Joshua L. Pughe\mbox{-}Sanford\textsuperscript{1,*}\quad
  Xuehao Ding\textsuperscript{1,*}\quad
  Jason J. Moore\textsuperscript{1,2,*}\quad
  Anirvan M. Sengupta\textsuperscript{1,3}\\[2pt]
  {\bfseries
  \hspace{-6mm}Charles Epstein\textsuperscript{1}\quad
  Philip Greengard\textsuperscript{1}\quad
  Dmitri B. Chklovskii\textsuperscript{1,2}}\\[2pt]
 \hspace{-5mm} \textsuperscript{1}Center for Computational Neuroscience, Flatiron Institute, Simons Foundation, New York, NY, USA\\
  \hspace{-6mm}\textsuperscript{2}Neuroscience Institute, NYU Langone Medical Center, New York, NY, USA\\
  \hspace{-5mm}\textsuperscript{3}Physics Department, Rutgers University, New Brunswick, NJ, USA\\[2pt]
  \hspace{-2mm}\texttt{\{jpughesanford, xding, cepstein, pgreengard, mitya\}@flatironinstitute.org,}\\
  \hspace{-2mm}\texttt{jason.moore@nyulangone.org,}\;
  \texttt{anirvans.physics@gmail.com}\quad
  \textsuperscript{*}\,Equal contribution
}

\maketitle
\section{Introduction to transfer operators}
This is only a concise primer, for more information, see~\cite{pavliotis2014stochastic,klus2024dynamical}.

\subsection{Discrete case}

We start from the discrete Markov chain for simplicity. Let $\mathcal X=\{x_1,\dots,x_N\}$. A one-step Markov transition is given by the kernel (matrix)
\[
P(y\mid x), \qquad \sum_{y} P(y\mid x)=1.
\]
Given an \emph{initial} distribution $\mu$ (row vector with $\mu(x)\ge 0$, $\sum_x \mu(x)=1$), the \emph{final} distribution after one step is
\[
\nu(y) =  \sum_{x} \mu(x)\, P(y\mid x).
\]

For an observable $f:\mathcal X\to\mathbb R$, the stochastic Koopman operator is
\[
(\mathcal K f)(x) = \sum_{y} P(y\mid x)\, f(y).
\]
Equip the input and output spaces with the weighted inner products
\[
\langle u,v\rangle_{\mu} = \sum_x u(x)\,v(x)\,\mu(x),
\qquad
\langle u,v\rangle_{\nu} = \sum_y u(y)\,v(y)\,\nu(y).
\]
The \emph{adjoint of $\mathcal K$ w.r.t.\ $(\mu,\nu)$} is the operator $\mathcal K^{\dagger}:L^2(\mu)\to L^2(\nu)$ defined by
\[
\langle \mathcal K f,\, g\rangle_{\mu} = \langle f,\, \mathcal K^{\dagger} g\rangle_{\nu}
\quad\text{for all } f,g.
\]
A direct computation yields, for any $y$ with $\nu(y)>0$,
\begin{equation}
(\mathcal K^{\dagger} g)(y)
= \frac{1}{\nu(y)} \sum_{x} \mu(x)\, P(y\mid x)\, g(x).
\label{eq:disc-adjoint}
\end{equation}
In matrix form, with $D_{\mu}=\mathrm{diag}(\mu)$ and $D_{\nu}=\mathrm{diag}(\nu)$,
\[
 \mathcal K^{\dagger} \;=\; D_{\nu}^{-1}\, P^{\!\top}\, D_{\mu}.
\]

The reverse-time kernel proposed by \cite{anderson1982reverse} is given by
\[
\tilde P(x\mid y) \;=\; \frac{\mu(x)\, P(y\mid x)}{\nu(y)}.
\]
Then \eqref{eq:disc-adjoint} can be rewritten as
\[
(\mathcal K^{\dagger} g)(y) \;=\; \sum_{x} \tilde P(x\mid y)\, g(x),
\]
i.e., $\mathcal K^{\dagger}$ is precisely the \emph{Koopman operator for the reverse-time Markov step} $y\mapsto x\sim \tilde P(\cdot\mid y)$ associated with pushing $\mu$ forward to $\nu=\mu P$.

\paragraph{Special cases.}
If $\mu=\nu=\pi$ is stationary and the chain is reversible ($\pi(x)P(y\mid x)=\pi(y)P(x\mid y)$), then $\tilde P(x\mid y)=P(x\mid y)$ and $\mathcal K^{\dagger}=\mathcal K$ (self-adjointness in $L^2(\pi)$).

\subsection{Continuous case}

Let $(\mathcal K_t)_{t\ge 0}$ be the Koopman semigroup with transition density $p_t(y\mid x)$:
\[
(\mathcal{K}_{t} f)(x) \;=\; \mathbb{E}[f(X_t)|X_0=x] \;=\; \int f(y)\, p_t(y\mid x)\, dy,
\]
\[
\mathcal{K}_t\mathcal{K}_s=\mathcal{K}_{t+s}.
\]
Given an initial distribution $\mu$ (measure or density), the pushed-forward distribution at time $t$ is
\[
\nu_t(dy) \;=\; \int \mu(dx)\, p_t(y\mid x)\, dy.
\]
With $\langle u,v\rangle_{\mu}=\int u(x)v(x)\,\mu(dx)$ and $\langle u,v\rangle_{\nu_t}=\int u(y)v(y)\,\nu_t(dy)$, the $(\mu,\nu_t)$-weighted adjoint $\mathcal{K}_{t}^{\dagger}:L^2(\mu)\to L^2(\nu_t)$ satisfies
\[
\langle \mathcal{K}_{t} f,\, g\rangle_{\mu} \;=\; \langle f,\, \mathcal{K}_{t}^{\dagger} g\rangle_{\nu_t},
\]
and is given by the reverse-time Koopman operator
\begin{equation}
\mathcal{K}_{t}^{\dagger} g(y) \;=\; \int g(x)\, \tilde p_t(x\mid y)\, dx,
\qquad
\tilde p_t(x\mid y) \;=\; \frac{\mu(x)\, p_t(y\mid x)}{\nu_t(y)}.
\label{eq:cts-adjoint}
\end{equation}
Thus the adjoint w.r.t.\ $(\mu,\nu_t)$ is the Koopman operator for the time-reversed process at horizon $t$ determined by $\mu$.

Consider an SDE $dX_t = a(X_t)\,dt + \sqrt{D}\, dW_t$. The generator of the Koopman operator can be calculated using Itô's formula as
\begin{equation}
\begin{aligned}
&\lim_{t\rightarrow 0}\frac{\mathcal{K}_tf(x)-f(x)}{t}\\
=& \mathbb{E}[\lim_{t\rightarrow 0}\frac{f(X_t)-f(x)}{t}|X_0=x]\\
=& a(x)\cdot \nabla f(x)+\tfrac12\, \nabla\cdot(D \nabla f(x))\\
=& \mathcal{L}^\dagger f(x).
\end{aligned}
\end{equation}

Thus, we obtain the equality $\mathcal{K}_{t} = e^{\mathcal L^\dagger t}$.

\section{The invariant measure}
Here we prove that the distribution given by Eq. (14, main text) is invariant under the action of the forward Kolmogorov operator, Eq. (11, main text):
\begin{equation}
\label{invariant}
\begin{aligned}
\mathcal{L}\rho_0({\mathbf x}) =&-\nabla\cdot [A\mathbf x \rho_0(\mathbf x)]+\frac{1}{2}\nabla\cdot [D\nabla \rho_0(\mathbf x)]\\
=&[-\Tr(A)+\mathbf x^\top A^\top \Sigma^{-1}\mathbf x - \frac{1}{2}\Tr(D\Sigma^{-1}) + \frac{1}{2}\mathbf x^\top \Sigma^{-1}D\Sigma^{-1}\mathbf x] \rho_0(\mathbf x)\\
=& [-\Tr(A) - \frac{1}{2}\Tr(D\Sigma^{-1})]\rho_0(\mathbf x)\\
=& -\frac{1}{2}\Tr(A-\Sigma A^\top \Sigma^{-1})\rho_0(\mathbf x)\\
=&0,
\end{aligned}
\end{equation}
where the Lyapunov equation, Eq. (15, main text),  was used to obtain the 3rd and the 4th line.

\section{Singular Spectrum of the SKO for OU processes}

In order to find the singular functions of the SKO, we search for the eigenfunctions of the \textbf{F}orward-backward and the \textbf{B}ackward-Forward operators:
\begin{align}
    \mathcal{F}_\tau &= e^{\mathcal{L}^\dagger\tau} \text{diag}(\rho_0)^{-1}e^{\mathcal{L}\tau}\text{diag}(\rho_0) & \mathcal{B}_\tau &= \text{diag}(\rho_0)^{-1}e^{\mathcal{L}\tau}\text{diag}(\rho_0) e^{\mathcal{L}^\dagger\tau}
\end{align}
where $\text{diag}(\rho_0)^{-1}e^{\mathcal{L}\tau}\text{diag}(\rho_0)$ is termed the reweighted Perron-Frobenius operator, which shares the domain with the SKO.

\subsection{Linear eigenfunctions}

We start by considering candidate eigenfunctions that are linear  w.r.t. the input state:
\begin{equation}
    q_\theta(\mathbf x) = \mathbf{\theta}^\top \mathbf x,
\end{equation}
where $\theta \in \mathbb{R}^n$. The actions of $\text{diag}(\rho_0)^{-1}\mathcal{L} \text{diag}(\rho_0)$ and $\mathcal{L}^\dagger$ on  $q_\theta (\mathbf x)$ are given by
\begin{equation}
\begin{aligned}
    &\text{diag}(\rho_0)^{-1}\mathcal{L} \text{diag}(\rho_0)q_\theta(\mathbf{x} )\\
    =& \rho_0(\mathbf{x})^{-1}[-\nabla\cdot(A \mathbf{x} \rho_0(\mathbf{x})\mathbf{\theta^\top x})+\frac{1}{2}D^{\alpha\beta}\partial_\alpha\partial_\beta(\rho_0(\mathbf x)\mathbf{\theta^\top x})]\\
    =& -\mathbf \theta^\top A \mathbf x - \mathbf x^\top \Sigma^{-1}D\mathbf \theta + \rho_0(\mathbf{x})^{-1}\mathbf{\theta^\top x} \mathcal{L}\rho_0(\mathbf x)\\
    =& (\Sigma^{-1}A\Sigma\mathbf{\theta})^\top\mathbf x,
\end{aligned}
\end{equation}
\begin{equation}
\begin{aligned}
&\mathcal{L}^\dagger q_\theta(\mathbf{x} )\\
= &(A\mathbf x)^\top\nabla(\mathbf\theta^\top\mathbf x)\\
= &(A^\top\mathbf\theta)^\top\mathbf x,
\end{aligned}
\end{equation}
where the Lyapunov equation, Eq. (15, main text),  and invariant density annihilation, Eq.\eqref{invariant}, were used. 

Then the actions of $\text{diag}(\rho_0)^{-1}e^{\mathcal{L} \tau} \text{diag}(\rho_0)$ and $e^{\mathcal{L}^\dagger \tau}$ on  $q_\theta (\mathbf x)$ are given by
\begin{equation}
\label{finite l reweighted}
    \begin{aligned}
    &\text{diag}(\rho_0)^{-1}e^{\mathcal{L} \tau} \text{diag}(\rho_0)q_\theta(\mathbf{x} )\\
    =&\text{diag}(\rho_0)^{-1}\sum_{i=0}^\infty \frac{\mathcal{L}^i\tau^i}{i!} \text{diag}(\rho_0)q_\theta(\mathbf{x} )\\
    =& (\sum_{i=0}^\infty \frac{\Sigma^{-1}A^i\tau^i\Sigma}{i!}\mathbf{\theta})^\top\mathbf x\\
    =& (\Sigma^{-1}e^{A \tau}\Sigma\mathbf{\theta})^\top\mathbf x,
\end{aligned}
\end{equation}

\begin{equation}
\label{finite l dagger}
\begin{aligned}
&e^{\mathcal{L}^\dagger \tau} q_\theta(\mathbf{x} )\\
=& \sum_{i=0}^\infty \frac{\mathcal{L}^{\dagger i}\tau^i}{i!}q_\theta(\mathbf{x} )\\
=& (\sum_{i=0}^\infty \frac{(A^\top)^{i}\tau^i}{i!}\mathbf \theta)^\top\mathbf x\\
=& (e^{A^\top \tau}\theta)^\top\mathbf x.
\end{aligned}
\end{equation}
Combining Eqs.\eqref{finite l reweighted}\eqref{finite l dagger}, denoting the matrices $\exp(A^\top \tau)\Sigma^{-1}\exp(A \tau)\Sigma$ and $\Sigma^{-1}\exp(A \tau)\Sigma \exp(A^\top \tau)$ by $P$ and $Q$, respectively, we have
\begin{equation}
    \mathcal{F}_\tau q_\theta(\mathbf x)=q_{P\mathbf{\theta}}(\mathbf x),\quad \quad
    \mathcal{B}_\tau q_\theta(\mathbf x)=q_{Q\mathbf{\theta}}(\mathbf x).
\end{equation}
Rewriting Eq. (18, main text) as 
\begin{align}
P \mathbf v_i = \lambda_i^2 \mathbf v_i,\quad \quad
     Q\mathbf u_i = \lambda_i^2 \mathbf u_i,
\end{align}
we obtain a set of linear eigenfunctions of the forward-backward and the backward-forward operators:
\begin{equation}
    \mathcal{F}_\tau q_{v_i}(\mathbf x)=\lambda_i^2 q_{v_i}(\mathbf x),\quad \quad
    \mathcal{B}_\tau q_{u_i}(\mathbf x)=\lambda_i^2 q_{u_i}(\mathbf x).
\end{equation}
So long as each set of eigenvectors are linearly independent, this method obtains all linear eigenfunctions $\mathcal{F}_\tau$ and $\mathcal{B}_\tau$. This is because the dimensionality of the subspace spanned by linear functions is equal to $n$, which equals the number of eigenfunctions found.

We note that our calculations apply regardless of the stability of the dynamics, i.e. whether the real parts of the eigenvalues of $\bf A$ are positive or negative. For stable $\bf A$, all $\{\lambda_i^2\}$ are less than one. The sign of the eigenfunction with the largest $\lambda_i^2$ partitions the state space into the minimally mixing 
coherent sets relative to the invariant density.
For unstable dynamics, an invariant density in the form of Eq. 14 of main text persists (ableit not normalizable) and all $\{\lambda_i^2\}$ are greater than one. In this setting, the sign of the eigenfunction with the smallest $\lambda_i^2$ partitions the state space into the minimally mixing coherent sets relative to the invariant density. Intuitively, one can view the unstable forward dynamics as a stable dynamics backwards in time of the dual variables \cite{anderson1982reverse, cvitanovic2012knowing, heninger2015neighborhoods}. 

\subsection{Quadratic eigenfunctions}
Here, we consider candidate eigenfunctions that are up to the second order in the input state:
\begin{equation}
    q_{M;c}(\mathbf x)=\mathbf{x}^\top M\mathbf x+c,
\end{equation}
where $M$ is an $n \times n$ symmetric matrix, $c$ is a scalar. It will soon become clear why we need an extra constant scalar rather than a linear term in the function. Since $\mathcal{L}c=\mathcal{L}^\dagger c=0$, $q_M(\mathbf x)$ will be used in the derivation below for brevity.\\

We first determine the actions of $\text{diag}(\rho_0)^{-1}\mathcal{L} \text{diag}(\rho_0)$ and $\mathcal{L}^\dagger$ on $q_M(\mathbf x)$.
\begin{equation}
\begin{aligned}
    &\text{diag}(\rho_0)^{-1}\mathcal{L} \text{diag}(\rho_0)q_M(\mathbf{x} )\\
    =& \rho_0(\mathbf{x})^{-1}[-\nabla\cdot(A \mathbf{x} \rho_0(\mathbf{x})\mathbf{x^\top}M\mathbf x)+\frac{1}{2}D^{\alpha\beta}\partial_\alpha\partial_\beta(\rho_0(\mathbf x)\mathbf{x^\top}M\mathbf x)]\\
    =& -2\mathbf x^\top M A \mathbf x - 2\mathbf x^\top \Sigma^{-1}DM\mathbf x + \Tr(DM)\\
    =& 2\mathbf x^\top\Sigma^{-1}A\Sigma M\mathbf x - 2\Tr(A\Sigma M)\\
    =& \mathbf x^\top (\Sigma^{-1}A\Sigma M + M\Sigma A^\top \Sigma^{-1})\mathbf x - 2\Tr(A\Sigma M),
\end{aligned}
\end{equation}
\begin{equation}
\begin{aligned}
&\mathcal{L}^\dagger q_M(\mathbf{x} )\\
= &(A \mathbf x)^\top\nabla(\mathbf x^\top M\mathbf x)+\frac{1}{2}D^{\alpha\beta}\partial_\alpha\partial_\beta (\mathbf x^\top M\mathbf x)\\
= &2\mathbf x^\top A^\top M \mathbf x - 2\Tr(A\Sigma M)\\
= & \mathbf x^\top (A^\top M + MA) \mathbf x - 2\Tr(A\Sigma M),
\end{aligned}
\end{equation}
where the Lyapunov equation and the fact $\mathcal{L}\rho_0 = 0$ have been used.

We noticed that
\begin{equation}
\begin{aligned}
&(\mathcal{L}^{\dagger})^2q_M(\mathbf x)=\mathbf x^\top ((A^\top)^2 M + 2A^\top MA + MA^2) \mathbf x + const.\\
&(\mathcal{L}^{\dagger})^3q_M(\mathbf x)=\mathbf x^\top ((A^\top)^3 M + 3(A^\top)^2MA + 3A^\top MA^2 + MA^3) \mathbf x + const.\\
&...\\
&(\mathcal{L}^{\dagger})^mq_M(\mathbf x) = \mathbf x^\top(\sum_{i=0}^mC_m^i(A^\top)^{m-i}MA^i)\mathbf x+ const.
\end{aligned}
\end{equation}
Thus, the action of $e^{\mathcal{L}^\dagger \tau}$ on $q_M (\mathbf x)$ is given by
\begin{equation}
\begin{aligned}
e^{\mathcal{L}^{\dagger}\tau}q_M(\mathbf x)=&\mathbf x^\top(\sum_{m=0}^\infty \frac{\tau^m}{m!}\sum_{i=0}^mC_m^i(A^\top)^{m-i}MA^i)\mathbf x+ const.\\
=& \mathbf x^\top(\sum_{i=0}^\infty \sum_{m=i}^\infty \frac{(A^\top)^{m-i}\tau^{m-i}}{(m-i)!}M\frac{A^{i}\tau^{i}}{i!})\mathbf x+ const.\\
=& \mathbf x^\top(e^{A^\top \tau}Me^{A \tau})\mathbf x + const.
\end{aligned}
\end{equation}
Similarly, the action of $\text{diag}(\rho_0)^{-1}e^{\mathcal{L} \tau} \text{diag}(\rho_0)$ on $q_M (\mathbf x)$ is given by
\begin{equation}
    \text{diag}(\rho_0)^{-1}\mathcal{L}^{m} \text{diag}(\rho_0)q_M(\mathbf x) = \mathbf x^\top(\sum_{i=0}^mC_m^i\Sigma^{-1}A^{m-i}\Sigma M \Sigma (A^\top)^{i}\Sigma^{-1})\mathbf x+ const.
\end{equation}
\begin{equation}
    \text{diag}(\rho_0)^{-1}e^{\mathcal{L} \tau} \text{diag}(\rho_0)q_M(\mathbf x) =\mathbf x^\top(\Sigma^{-1}e^{A \tau}\Sigma M \Sigma e^{A^\top \tau}\Sigma^{-1})\mathbf x+ const.
\end{equation}
Finally we have the actions of $\mathcal{F}_\tau$ and $\mathcal{B}_\tau$ on $q_M(\mathbf x)$:
\begin{equation}
    \mathcal{F}_\tau q_M(\mathbf x)=\mathbf x^\top(PMP^\top)\mathbf x + const.
\end{equation}
\begin{equation}
    \mathcal{B}_\tau q_M(\mathbf x)=\mathbf x^\top(QMQ^\top)\mathbf x + const.
\end{equation}
With careful choice of the constant term, we obtain a set of quadratic eigenfunctions
\begin{equation}
\begin{aligned}
    &\mathcal{F}_\tau q_{v_iv_i^\top}(\mathbf x)=\lambda_i^4q_{v_iv_i^\top}(\mathbf x),\\
    &\mathcal{F}_\tau q_{v_iv_j^\top+v_jv_i^\top}(\mathbf x)=\lambda_i^2\lambda_j^2 q_{v_iv_j^\top+v_jv_i^\top}(\mathbf x),
\end{aligned}
\end{equation}
\begin{equation}
\begin{aligned}
    &\mathcal{B}_\tau q_{u_iu_i^\top}(\mathbf x)=\lambda_i^4q_{u_iu_i^\top}(\mathbf x),\\
    &\mathcal{B}_\tau q_{u_iu_j^\top+u_ju_i^\top}(\mathbf x)=\lambda_i^2\lambda_j^2 q_{u_iu_j^\top+u_ju_i^\top}(\mathbf x).
\end{aligned}
\end{equation}
Since the number of eigenfunctions ($\frac{n(n+1)}{2}$) is equal to the dimensionality of the subspace spanned by $n\times n$ symmetric matrices, we obtained all quadratic eigenfunctions. For attractive dynamics, $0<\lambda_i^2<1$, so their eigenvalues are less than the greatest eigenvalue of linear functions. For repulsive dynamics, $\lambda_i^2>1$, so their eigenvalues are greater than the smallest eigenvalue of linear functions. 

\subsection{General situation}
On $\mathbb{R}^n$ consider a symmetric k-linear form $M_k:(\mathbb{R}^{n})^{\otimes k}\rightarrow \mathbb{R}$ such that $M_k(v_1\otimes v_2 \otimes ... v_k)=M_k(v_{\sigma(1)}\otimes v_{\sigma(2)} \otimes ... v_{\sigma(k)})$ for any permutation $\sigma \in S_k$. Its matrix representation is given by $M_k^{I_1I_2...I_k}=M_k(e_{I_1}\otimes e_{I_2}\otimes ... e_{I_k})$ where $\{e_{i}:i=1,2...n\}$ is the standard orthonormal basis of $\mathbb{R}^n$. Then we can consider $k$-th order tentative eigenfunctions of the form:
\begin{equation}
    q^k_M(\mathbf x) = M_k(\mathbf x^{\otimes k}) + O(k-2),
\end{equation}
where $O(k-2)$ denotes a (multivariate) polynomial up to the $(k-2)$-th order. Similar to the quadratic case, the eigenvalue is solely determined by the leading term. Thus, it is sufficient to analyze the actions of operators on $M_k$.

The actions of $\text{diag}(\rho_0)^{-1}\mathcal{L} \text{diag}(\rho_0)$ and $\mathcal{L}^\dagger$ on $M_k(\mathbf x^{\otimes k})$ are given by
\begin{equation}
\begin{aligned}
    &\text{diag}(\rho_0)^{-1}\mathcal{L} \text{diag}(\rho_0)M_k(\mathbf{x}^{\otimes k} )\\
    =& \rho_0(\mathbf{x})^{-1}[-\nabla\cdot(A \mathbf{x} \rho_0(\mathbf{x})M_k(\mathbf{x}^{\otimes k} ))+\frac{1}{2}\nabla\cdot(D\nabla \rho_0(\mathbf x)M_k(\mathbf{x}^{\otimes k} ))]\\
    =& -\langle A\mathbf x,\nabla M_k(\mathbf{x}^{\otimes k})\rangle+\langle\nabla \log \rho_0(\mathbf x), D\nabla M_k(\mathbf{x}^{\otimes k})\rangle+\frac{1}{2}\nabla\cdot(D\nabla M_k(\mathbf{x}^{\otimes k} ))\\
    =&\langle-(A+D\Sigma^{-1})\mathbf{x}, \nabla M_k(\mathbf{x}^{\otimes k} )\rangle + O(k-2)\\
    =&\langle\Sigma A^\top \Sigma^{-1}\mathbf x, \nabla M_k(\mathbf{x}^{\otimes k} )\rangle + O(k-2)\\
    =& \sum_{i=1}^k M_k\times_i (\Sigma^{-1}A\Sigma)(\mathbf{x}^{\otimes k})+ O(k-2),
\end{aligned}
\end{equation}
\begin{equation}
\begin{aligned}
&\mathcal{L}^\dagger M_k(\mathbf{x}^{\otimes k} )\\
= &\langle A \mathbf x,\nabla M_k(\mathbf{x}^{\otimes k} )\rangle+\frac{1}{2}\nabla\cdot(D\nabla M_k(\mathbf{x}^{\otimes k} ))\\
= &\sum_{i=1}^k M_k\times_i A^\top(\mathbf{x}^{\otimes k}) +O(k-2),\\
\end{aligned}
\end{equation}
where the Lyapunov equation and the fact $\mathcal{L}\rho_0 = 0$ have been used. $\times_i$ denotes tensor mode product: $(T\times_n A)_{i_1,...,i_{n_1},j,i_{n+1},...,i_k}\equiv \sum_{i_n}A_{ji_n}T_{i_1,...,i_n,...,i_k}$ satisfying the following algebra.
\begin{equation}
\begin{aligned}
    &T\times_iA \times_iB = T\times_i(BA),\\
    &T\times_iA \times_jB = T\times_jB \times_iA \quad i\neq j.
\end{aligned}
\end{equation}

Then we can determine the actions of finite-time operators on $M_k(\mathbf{x}^{\otimes k} )$.
\begin{equation}
    (\mathcal{L}^{\dagger})^m M_k(\mathbf{x}^{\otimes k})=\sum_{i_m=1}^k...\sum_{i_2=1}^k\sum_{i_1=1}^k M_k\times_{i_1} A^\top \times_{i_2} A^\top ... \times_{i_m} A^\top(\mathbf{x}^{\otimes k}) +O(k-2),
\end{equation}
\begin{equation}
\begin{aligned}
&e^{\mathcal{L}^\dagger \tau}M_k(\mathbf{x}^{\otimes k})\\
=&\sum_{m=0}^\infty \frac{\tau^m(\mathcal{L}^{\dagger})^m}{m!}M_k(\mathbf{x}^{\otimes k})\\
=& \sum_{\{l_1,l_2,...l_{k-1}\}=0}^\infty\sum_{m=\sum_{i=1}^{k-1} l_i}^\infty \frac{\tau^m}{m!}{m \choose l_1, l_2, \ldots, l_{k-1}, m-\sum_{i=1}^{k-1}l_i}\cdot\\
&M_k\times_1 (A^\top)^{l_1}\times_2 (A^\top)^{l_2}...\times_{k-1} (A^\top)^{l_{k-1}}\times_k (A^\top)^{m-\sum_{i=1}^{k-1}l_i}(\mathbf{x}^{\otimes k})+O(k-2)\\
=& M_k\times_1 e^{A^\top \tau}\times_2 e^{A^\top \tau}...\times_k e^{A^\top \tau}(\mathbf{x}^{\otimes k})+O(k-2).
\end{aligned}    
\end{equation}
Similarly,
\begin{equation}
\begin{aligned}
&\text{diag}(\rho_0)^{-1}e^{\mathcal{L} \tau} \text{diag}(\rho_0)M_k(\mathbf x^{\otimes k})\\
=&M_k\times_1 (\Sigma^{-1}e^{A \tau}\Sigma)\times_2 (\Sigma^{-1}e^{A \tau}\Sigma)...\times_k (\Sigma^{-1}e^{A \tau}\Sigma)(\mathbf{x}^{\otimes k})+O(k-2).
\end{aligned}
\end{equation}
Then the actions of $\mathcal{F}_\tau$ and $\mathcal{B}_\tau$ on $M_k(\mathbf x^{\otimes k})$ follow as
\begin{equation}
    \mathcal{F}_\tau M_k(\mathbf x^{\otimes k})=M_k\times_1 P\times_2 P...\times_k P(\mathbf{x}^{\otimes k})+O(k-2),
\end{equation}
\begin{equation}
    \mathcal{B}_\tau M_k(\mathbf x^{\otimes k})=M_k\times_1 Q\times_2 Q...\times_k Q(\mathbf{x}^{\otimes k})+O(k-2).
\end{equation}

We therefore obtain a set of $k$-th-order eigenfunctions of $\mathcal{F}_\tau$ and $\mathcal{B}_\tau$:
\begin{equation}
    \mathcal{F}_\tau q^k_{\sum_{\sigma\in S_k}v_{I_{\sigma(1)}}\otimes v_{I_{\sigma(2)}}\otimes ... v_{I_{\sigma(k)}}}(\mathbf{x}^{\otimes k})=(\prod_{i=1}^k \lambda_{I_i}^2) q^k_{\sum_{\sigma\in S_k}v_{I_{\sigma(1)}}\otimes v_{I_{\sigma(2)}}\otimes ... v_{I_{\sigma(k)}}}(\mathbf{x}^{\otimes k}),
\end{equation}
\begin{equation}
    \mathcal{B}_\tau q^k_{\sum_{\sigma\in S_k}u_{I_{\sigma(1)}}\otimes u_{I_{\sigma(2)}}\otimes ... u_{I_{\sigma(k)}}}(\mathbf{x}^{\otimes k})=(\prod_{i=1}^k \lambda_{I_i}^2) q^k_{\sum_{\sigma\in S_k}u_{I_{\sigma(1)}}\otimes u_{I_{\sigma(2)}}\otimes ... u_{I_{\sigma(k)}}}(\mathbf{x}^{\otimes k}).
\end{equation}
Again, these eigenvalues are less than the greatest eigenvalue of linear eigenfunctions for attractive dynamics and greater than the smallest eigenvalue of linear eigenfunctions for repulsive dynamics.

\section{Orthogonality analysis of singular vectors for OU processes}
\label{eig_to_sing}
    This section proves that asymptotically the top predictive mode converges to the top left eigenvector and the bottom retrospective mode converges to the bottom left eigenvector. This helps us understand the 2d saddle-point example in Fig. 1 and provides analytical support for the orthogonality analysis of data in Figs. 2 and 3. 

\subsection{Top eigenvector}

We start from the definitions of eigenvectors of transfer operators:
\begin{align}
\label{sigma}
e^{A^\top\tau}\Sigma^{-1}e^{A\tau}\Sigma \mathbf v_i = \lambda_i^2 \mathbf v_i,\quad \quad
     \Sigma^{-1}e^{A\tau}\Sigma e^{A^\top\tau} \mathbf u_i = \lambda_i^2 \mathbf u_i.
\end{align}
Notice that
\begin{equation}
\begin{aligned}
    &(e^{A^\top\tau}\Sigma^{-1}e^{A\tau}\Sigma)^\top \Sigma \mathbf v_i\\
    =&\Sigma e^{A^\top\tau}\Sigma^{-1} e^{A\tau}\Sigma\mathbf v_i\\
    =&\lambda_i^2\Sigma\mathbf v_i.
\end{aligned}
\end{equation}
Thus, $\{\mathbf v_i, \Sigma\mathbf v_i\}$ is a pair of right and left eigenvectors of $e^{A^\top\tau}\Sigma^{-1}e^{A\tau}\Sigma$. Similarly, $\{\mathbf u_i, \Sigma\mathbf u_i\}$ is a pair of right and left eigenvectors of $\Sigma^{-1}e^{A\tau}\Sigma e^{A^\top\tau}$. Therefore, $\mathbf u_i^\top \Sigma\mathbf u_j=\mathbf v_i^\top\Sigma\mathbf v_j=0 \text{ for } i\neq j$.\\

Also notice that
\begin{equation}
\begin{aligned}
    &e^{A^\top \tau}(\Sigma^{-1}e^{A\tau}\Sigma e^{A^\top\tau} \mathbf u_i)\\
    =&(e^{A^\top \tau}\Sigma^{-1}e^{A\tau}\Sigma) e^{A^\top\tau} \mathbf u_i\\
    =&\lambda_i^2 e^{A^\top \tau}\mathbf u_i.
\end{aligned}
\end{equation}
Thus, $\mathbf v_i\propto e^{A^\top \tau}\mathbf u_i$. Under proper normalization, these results can be summarized by
\begin{align}
    V^\top\Sigma V=Id,\quad U^\top\Sigma U=Id,\quad e^{A^\top \tau}U= V\Lambda,
\end{align}
where $U[:,i]\coloneq \mathbf u_i$, $V[:,i]\coloneq \mathbf v_i$, $\Lambda_{ij}\coloneq \delta_{ij}\lambda_i$. Thus, we have
\begin{align}
    &e^{A^\top \tau}=V\Lambda U^{-1}=V\Lambda U^\top \Sigma,\\
    &e^{A \tau}=\Sigma U \Lambda V^\top=\sum_i \lambda_i\Sigma \mathbf{u}_i\mathbf{v}_i^\top.
    \label{svd of e^(At)}
\end{align}
On the other hand, we have the eigen-decomposition
\begin{align}
    e^{A\tau}=\sum_i e^{a_i\tau}\mathbf r_i \mathbf l_i^\top,
    \label{eigen-decomposition of e^(At)}
\end{align}
where $a_i$ is the eigenvalue of $A$, $\mathbf r_i$ and $\mathbf l_i$ are right and left eigenvectors of $A$, respectively. We consider the asymptotic limit as $\tau\rightarrow \infty$, where the leading rank-1 term of $e^{A\tau}$ dominates. Comparing the leading terms of Eq. \eqref{svd of e^(At)} and Eq. \eqref{eigen-decomposition of e^(At)}, we have
\begin{align}
    \mathbf v_1 \propto \mathbf l_1, \quad \Sigma \mathbf u_1 \propto \mathbf r_1.
\end{align}
Thus, the top predictive mode is orthogonal to every eigen-direction except the most repulsive direction. In particular, $\mathbf v_1$ is orthogonal to the stable manifold.

\subsection{Bottom eigenvector}

Notice that under the transformation $A\rightarrow - A^\top $, we have
\begin{align}
\begin{aligned}
    e^{A\tau}\rightarrow & [e^{A\tau}]^{-\top} \\
    = & [\Sigma U \Lambda V^\top]^{-\top} \\
    = & U \Lambda^{-1} V^\top\Sigma\\
    = & \sum_i \lambda_i^{-1}\mathbf{u}_i(\Sigma\mathbf{v}_i)^\top.
\end{aligned}
\end{align}
And the eigen-decomposition becomes
\begin{align}
    e^{-A^\top\tau}=\sum_i e^{-a_i\tau}\mathbf l_i \mathbf r_i^\top,
    \label{eigen-decomposition of e^(-A^Tt)}
\end{align}
Similar to the first section, by matching the leading rank-1 term under the $\tau\rightarrow \infty$ limit, we have
\begin{align}
    \Sigma\mathbf v_n \propto \mathbf r_n, \quad  \mathbf u_n \propto \mathbf l_n.
\end{align}
Thus, the bottom retrospective mode is orthogonal to every eigen-direction except the most attractive direction. In particular, $\mathbf u_n$ is orthogonal to the unstable manifold.

\section{Galerkin approximation of transfer operators \cite{klus2024dynamical}}
Real neurons are exposed to only a finite number of channels or synaptic inputs. Mathematically, this means that it can only have access to a finite collection of observables, which we write as the Galerkin projection, referring to the projection of the SKO onto the linear span of these accessible measurement functions.
Let
\begin{equation}
\mathcal{F} = \{ f({\bf x}) \ |\ f({\bf x}) = \sum_{i=1}^d c_i \phi_i({\bf x})= \vec{c}\cdot \vec{\phi}({\bf x})\},
\end{equation}
denote the linear span of a finite set of basis functions $\{\phi_i\}_{i=1}^d$.
Since each $f \in \mathcal{F}$ is uniquely determined by the coefficient vector $\vec{c} \in \mathbb{R}^d$, the space $\mathcal{F}$ is naturally isomorphic to $\mathbb{R}^d$. Given some inner product, any linear operator $\mathcal{A}: \mathcal{F}\to\mathcal{F}$ admits a matrix representation $A:\mathbb{R}^d\to\mathbb{R}^d$ satisfying
\begin{align}\label{eq:galerkin}
[\mathcal{A}\ \ {\vec c}\cdot \vec{\phi}]({\bf x}) = (A \vec{c})\cdot \vec{\phi}({\bf x}),
\end{align}
where $A = (\Sigma)^{-1}\Sigma_\mathcal{A} $ and 
\begin{align}
\Sigma_{ij} &= \langle \phi_i, \phi_j \rangle, & [\Sigma_\mathsf{A}]_{ij} &= \langle \phi_i, \mathcal{A}\phi_j \rangle.
\end{align}
For functions $f = \vec{c}_f \cdot \vec{\phi}$ and $g = \vec{c}_g \cdot \vec{\phi}$ in $\mathcal{F}$, the inner product reduces to
\begin{align}\label{eq:innerprod}
\langle f,g \rangle &= \vec{c}_f\cdot  \Sigma \vec{c}_g.
\end{align}
This motivates the Galerkin projection of the SKO, which is a linear operator describing the application of $\mathcal{K}_\tau$, followed by projection onto $\mathcal{F}$. This operator maps $\mathcal{F}$ to itself, and admits a matrix representation, \cite{koltai2018,klus2024dynamical}
\begin{align}
  K_\tau &= (\Sigma)^{-1}\Sigma_{\mathcal{K}_\tau},
\end{align} 
where $\phi_i$ are assumed to be linearly independent such that $\Sigma$ is invertible. To specify $K_\tau$ entirely, one must define the inner product.

For a statistically stationary process with invariant density $\rho({\bf x})$, a natural inner product is defined as
\begin{align}
\langle f, g \rangle = \int_\mathcal{X} f({\bf x}) g({\bf x})\rho({\bf x})d{\bf x} = \mathbb{E}[f({\bf X}) g({\bf X})] \approx \frac{1}{S}\sum_{s=1}^S f({\bf X}(t_s)) g({\bf X}(t_s)).
\end{align}
The final approximation highlights the motivation for this choice in a data-driven context: the inner product can be computed directly from samples ${\bf X}(t_s) \sim \rho({\bf x})$ drawn from the process. 
Using this inner product 
\begin{align}
    [\Sigma]_{ij} \approx \frac{1}{S}\sum_{s=1}^S \phi_i({\bf X}(t_s)) \phi_j({\bf X}(t_s)),\quad\quad
    [\Sigma_{\mathsf{K}_\tau}]_{ij} \approx \frac{1}{S}\sum_{s=1}^S \phi_i({\bf X}(t_s)) \phi_j({\bf X}(t_s+\tau))
\end{align}
describe the correlation between measurement functions $i$ and $j$ measured instantaneously, or with delay $\tau$, respectively. Both $\Sigma_{\mathcal{K}_\tau}$ and $\Sigma^{-1}$ are simple to compute in biologically plausible online schemes.

\section{Equivalence between CCA and coherent sets for Stable OU Process}
\label{CCA_OU}
    This Section proves that the analytical approach to  coherent set projection by finding the subdominant singular function is equivalent to the data-driven Galerkin projection \cite{klus2024dynamical}, which is also equivalent to information bottleneck and canonical correlation analysis \cite{chechik2003information}, for stable OU processes.

\begin{proposition}
For stable OU processes, the theoretically predicted filters agree exactly with the CCA solution. That is,
\begin{equation}
    C_0^{-1}C_\tau C_0^{-1}C_{-\tau}=e^{A^\top \tau}\Sigma^{-1}e^{A \tau}\Sigma, \quad
    C_0^{-1}C_{-\tau} C_0^{-1}C_{\tau}=\Sigma^{-1}e^{A \tau}\Sigma e^{A^\top \tau}.
    \label{equivalence}
\end{equation}
\end{proposition}
\begin{proof}
The solution of an OU process is given by
\begin{equation}
    X_t=\int_{-\infty}^t  e^{A(t-s)}\xi(s)ds.
\end{equation}
We can then directly calculate the covariance matrices as follows.
\begin{equation}
\begin{aligned}
    C_0& \coloneq \mathbb{E}(X_tX_{t}^{\top})\\
    &=\int_{-\infty}^t \int_{-\infty}^{t}e^{A(t-s)}\mathbb{E}[\xi(s)\xi(s')^\top] e^{A^\top(t-s')}dsds'\\
    &=\int_{-\infty}^t \int_{-\infty}^{t}e^{A(t-s)}D\delta(s-s') e^{A^\top(t-s')}dsds'\\
    &=\int_{-\infty}^t e^{A(t-s)}D e^{A^\top(t-s)}ds\\
    &=-\int_{-\infty}^t e^{A(t-s)}(A\Sigma+\Sigma A^\top) e^{A^\top(t-s)}ds\\
    &=\int_{-\infty}^t \frac{d[e^{A(t-s)}\Sigma e^{A^\top(t-s)}]}{ds}ds\\
    &=\Sigma,
\end{aligned}
\end{equation}

\begin{equation}
\begin{aligned}
    C_\tau&\coloneq\mathbb{E}(X_tX_{t+\tau}^{\top})\\
    &=\int_{-\infty}^t \int_{-\infty}^{t+\tau}e^{A(t-s)}\mathbb{E}[\xi(s)\xi(s')^\top] e^{A^\top(t+\tau-s')}dsds'\\
    &=C_0 e^{A^\top \tau}\\
    &=\Sigma e^{A^\top \tau},
\end{aligned}
\end{equation}

\begin{equation}
\begin{aligned}
    C_{-\tau}&\coloneq\mathbb{E}(X_tX_{t-\tau}^{\top})\\
    &=\int_{-\infty}^t \int_{-\infty}^{t-\tau}e^{A(t-s)}\mathbb{E}[\xi(s)\xi(s')^\top] e^{A^\top(t-\tau-s')}dsds'\\
    &=e^{A \tau}C_0 \\
    &=e^{A \tau}\Sigma.
\end{aligned}
\end{equation}
From these equalities one can straightforwardly derive Eq. \eqref{equivalence}, showing that, for stable OU processes, our analytical method is equivalent to the data-driven method for finding the subdominant eigenfunction.
\end{proof}

\section{Experimental Data Analysis}
\subsection{Dead Leaves Model Simulation}
The plateaued stimulus used in Figure 2A was generated by concatenating $tanh$ functions with varying amplitude, time constant, and bias and sampled at unitary time steps. In Figure 2, the time on the x axis is shown in the time steps of the dataset. To construct the lag vectors necessary to estimate the past and future temporal filters, the dataset was sampled by a sliding window 10 time steps long. The mean-centered lag vectors were processed by CCA. 









\subsection{Tufted/Mitral Cell Dataset}
Data were taken from \cite{gupta2015olfactory}. Briefly, anesthetized rats were presented with time-varying stimuli of multiple odors while tufted/mitral cells were recorded with extracellular tetrodes. For each odor-cell pair, an odor kernel best describing the response of a cell to a brief odor pulse was estimated. Kernels were estimated at 20 Hz with a total length of 2 seconds. 

Kernels were parametrized as Gabor functions of the form
\[K(t) = a e^{(\frac{-(t-\tau)}{w})^2} cos(2\pi(f+t-\tau)+\phi)\]
with $a$, $\tau$, $\omega$, $f$, and $\phi$ as free parameters.

Here, we further set all kernel values with a magnitude below 
$0.005$ to zero (to remove noise) and normalized all kernels to have unit $\ell^2$ norm. 

\subsection{Retinal Ganglion Cell Dataset}
Data were taken from \cite{KastnerBaccus2011}. Briefly, retinal ganglion cells were recorded using an electrode array during display of high contrast gaussian white noise. Spatio-temporal receptive fields (STRFs) were estimated using reverse correlation. For each cell, we extracted its temporal receptive field by computing a rank-1 decomposition using the SVD. This rank-1 approximation accounted for approximately 60\% of the variance of the STRFs. Temporal kernels were estimated at 200 Hz with a total length of 0.5 seconds. Temporal receptive fields were then parameterized and normalized using the same procedure as above (Tufted/Mitral Cell Dataset).

\begin{figure}[h!]
  \centering
  \includegraphics[width=0.8\textwidth]{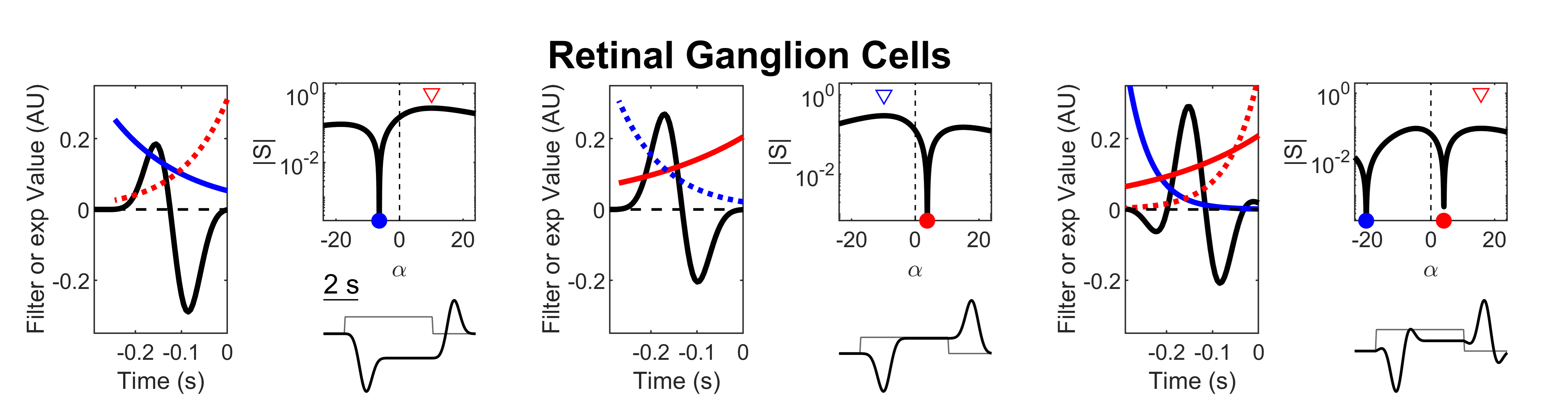}
  \vspace{-1mm}
  \caption{Sample temporal receptive fields from three retinal ganglion cells in the salamander retina. As in Figure 3, the black lines indicate individual linear filters of different neurons, the solid colored lines indicate exponential that are orthogonal to the filters, and the dotted colored lines indicate the exponentials that are most aligned with them. The bottom inset shows the convolution of the filter with a step pulse.}
  \vspace{-5mm}
\end{figure}

\subsection{Identifying Orthogonal Exponentials}
For each temporal receptive field in Figs. 2, 3, and S1, we numerically computed the cosine similarity between the filter and a normalized exponential. 

The cosine similarity S between two column vectors was computed as
$$S({\bf x},{\bf y}) = \frac{{\bf x}^\top {\bf y}}{\|{\bf x}\|\|{\bf y}\|}
$$
where $\|{\bf u}\|$ is the $\ell^2$ norm of the temporally discretized filter ${\bf u}\in \mathbb{R}^n$, where $n=40$ for the Tufted/Mitral cells and $n=100$ for the Retinal Ganglion cells. 

Exponentials had the form
\[E(t) = ce^{\alpha t}\]
where c was chosen such that the $L_2$ norm defined over the range of non-zero values of the temporal receptive field was 1. For tufted/mitral cells, $\alpha$ was swept between -10 and 10. For retinal ganglion cells, $\alpha$ was swept between -30 and 30. The zero-crossing of the cosine similarity as a function of $\alpha$ was identified as the  $\alpha$ yielding an orthogonal exponential. The computation for all 1549 neurons (204 Mitral/Tufted, 1345 RGCs) took less than 1 minute on a standard PC laptop (Intel i7-1260P, 16GB RAM).

\begin{figure}[h!]
  \centering
  \includegraphics[width=0.8\textwidth]{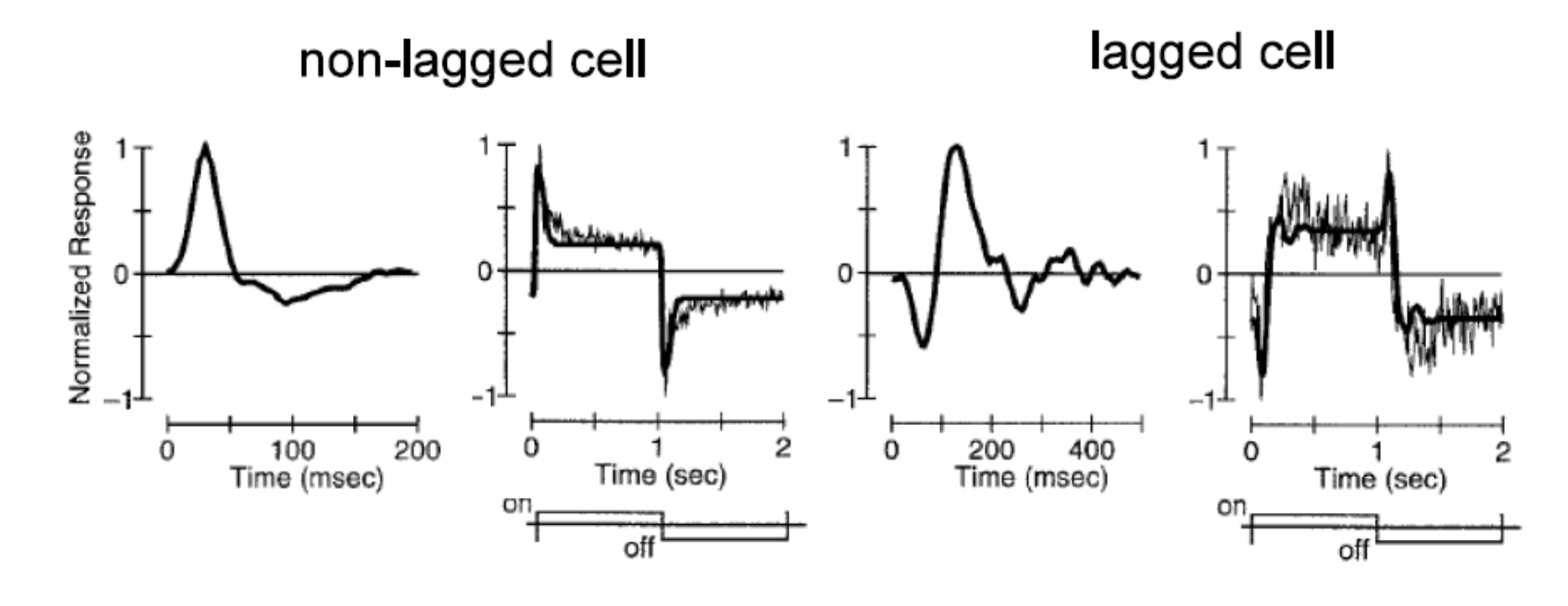}
  \vspace{-1mm}
  \caption{Sample temporal receptive fields and step responses from non-lagged and lagged cells of the cat LGN~\cite{cai1997spatiotemporal}. When comparing the temporal receptive fields with those in Fig. 3 and Fig. S1, note that the time axis is inverted. Step response lag accounts for the naming of the ceslls. 
  }
  \vspace{-5mm}
\end{figure}



\printbibliography